\documentclass[journal]{IEEEtran}

\usepackage{graphicx}
\usepackage{subfigure}
\usepackage{epstopdf}
\usepackage[cmex10]{amsmath}
\usepackage{amssymb}
\usepackage{amsthm}
\newtheorem{proposition}{Proposition}
\usepackage{wasysym}
\usepackage{relsize}
\usepackage{algorithm}
\usepackage{algorithmic}
\usepackage{array}
\usepackage{cite}
\usepackage{color}
\usepackage{url}
\usepackage{bm}
\usepackage{enumerate}

\usepackage{epsfig,latexsym}

\usepackage{booktabs}
\usepackage{makecell}

\begin{document}
	\title{Finite-Aperture Modeling of Pattern-Reconfigurable Antennas for Multiuser Beamforming}
	
	
	\author{Kangjian~Chen,~\IEEEmembership{Member,~IEEE},  Han Qing Yang,~\IEEEmembership{Member,~IEEE}, and Hang Wong,~\IEEEmembership{Fellow,~IEEE}
	\thanks{This work has been submitted to the IEEE for possible publication. Copyright may be transferred without notice, after which this version may no longer be accessible.}
 		\thanks {Kangjian~Chen is with the School of Electronic and Optical Engineering, Nanjing University of Science and Technology, Nanjing 210094, China (e-mail: chenkj@njust.edu.cn). }
		\thanks {Han Qing Yang and Hang Wong are with the State Key Laboratory of Terahertz and Millimeter Waves, Department of Electrical Engineering, City University of Hong Kong, Hong Kong (e-mail: hanqyang@cityu.edu.hk,  hang.wong@cityu.edu.hk). } 
	}
	{}
	
	\maketitle
	
	%
	%
	
	
\begin{abstract}
This paper establishes a finite-aperture model for pattern-reconfigurable antennas (PRAs) and exploits its aperture-constrained Fourier representation to develop an efficient multiuser beamforming framework. First, we develop an ideal model by representing the complex radiation pattern using a Fourier series. By applying the Jacobi--Anger expansion, we relate the antenna aperture to the effective Fourier order and derive aperture-dependent power bounds for individual Fourier coefficients. Second, we employ a practical multiport network model and design three representative reconfigurable pixel antennas to verify the predicted effective Fourier orders and per-order power constraints. We then formulate weighted sum-rate maximization for multiuser downlink transmission under both models. Using bounded Fourier coefficients as auxiliary variables, we develop a coefficient-bound penalty-based weighted minimum mean-square error algorithm that separates continuous beamforming from discrete switch-state optimization. We further propose a one-bit neighborhood search that exploits the Sherman--Morrison identity to accelerate the discrete optimization. Simulation results validate the physical relevance of the ideal model and demonstrate multiuser sum-rate gains and reduced optimization runtime for densely switched antennas. The proposed model provides a physically grounded basis for PRA-enabled communication analysis and design, while the framework accommodates different beamforming objectives and PRA implementations.
\end{abstract}

\section{Introduction}
Multi-antenna technology has long been a cornerstone of modern wireless communications~\cite{Zhu2026MovableReconfigurable}. From first-generation to fifth-generation wireless systems, antenna architectures have evolved from single-antenna systems to multiple-input multiple-output (MIMO) and massive MIMO. More recently, this evolution has continued toward extremely large-scale MIMO~\cite{Chen2024XLMIMO}. By exploiting array, diversity, and spatial-multiplexing gains, these technologies have substantially improved transmission reliability and spectral efficiency. Nevertheless, these advances have been achieved primarily by increasing the number of antenna elements, while antenna configurations, including array geometries and antenna radiation characteristics, have remained largely fixed. In future wireless networks,  increasingly complex propagation environments and more demanding communication requirements will make it difficult for such fixed antenna configurations to fully meet system needs. Therefore, in addition to further increasing the array scale, it is also important to enhance antenna flexibility~\cite{Ying2024RmMIMOOverview,Wang2026FlexibleMIMO,Ning2025MAArchitectures}.

Reconfigurable antennas (RAs) have emerged as a promising approach to improving antenna flexibility by introducing controllable changes in antenna configurations, such as spatial position, orientation, and intrinsic radiation characteristics~\cite{Ma2026RASurvey,TWCLMZ,TWC25YSJ}. By dynamically controlling these configurations, RAs extend system adaptation from the signal domain to the antenna domain, allowing the antenna configuration to adapt to varying propagation conditions and communication requirements. Their considerable potential has attracted increasing attention from the wireless communications community, leading to extensive research on various forms of RAs~\cite{Zheng2026RotatableAntenna,Shao2025SixDMA,chen2025dbraa,TWCLMZ}. Among these, position reconfiguration and pattern reconfiguration represent two important reconfiguration mechanisms.

Position reconfiguration changes the positions at which the electromagnetic field is excited. By adjusting the antenna positions, position-dependent phase variations are introduced across the propagation paths, thereby changing their coherent superposition and reshaping the effective channel. When position reconfiguration is jointly optimized with baseband and radio frequency (RF) signal processing, the resulting effective channel variations provide additional degrees of freedom to improve transmission performance, suppress interference, and enhance estimation accuracy. Representative antenna technologies based on position reconfiguration include movable antennas (MA)~\cite{Zhu2024MAMultiuser}, fluid antenna systems (FAS)~\cite{Wong2021FAS}, and pinching-antenna systems (PASS)~\cite{Ding2025PinchingAntenna,chen2026hybrid}. MA physically repositions antennas within a prescribed region to obtain favorable channel conditions~\cite{Zhu2026MATutorial,chen2025remaa}. FAS realizes position flexibility by selecting antennas among densely distributed candidate ports within a compact region~\cite{New2025FASTutorial,Zhang2026FASISACURA,Zhang2026FASJADCE}. PASS creates radiating points at configurable locations along a dielectric waveguide~\cite{Liu2026PASSOverview}. Existing studies have investigated these technologies from multiple perspectives, including performance analysis, channel estimation, and joint position and beamforming optimization.

Pattern reconfiguration modifies the intrinsic complex angular response of an antenna by redistributing the surface currents over its radiating structure. Such redistribution is realized using tunable RF components, such as positive-intrinsic-negative (PIN) diodes and varactors~\cite{Christodoulou2012Reconfigurable,Liu2026RAArrays,TWCLMZ}. Research in this direction includes conventional pattern-reconfigurable antennas (PRAs) and implementations of the same pattern-reconfiguration principle developed within the FAS framework~\cite{Zhang2025PRAFAS}. Although these antenna technologies may employ different physical structures, their channel-adaptation capability is primarily derived from their controllable angular responses. An appropriate model of pattern reconfiguration is essential for communication-system analysis and design. Existing communication-oriented PRA models can be broadly divided into hardware-independent ideal models and hardware-aware practical models. Ideal models represent the radiation pattern using continuous functions or orthogonal basis expansions, such as spherical harmonics~\cite{Ying2025Reconfigurable,Zheng2026TriHybrid}. Practical models instead represent the radiation pattern using precharacterized pattern codebooks or physical models, such as multiport networks~\cite{Shen2026AntennaCoding}. Based on these models, existing studies have mainly investigated pattern selection or joint pattern and beamforming optimization.

In this paper, we focus on PRAs because their electronically controlled electromagnetic reconfigurability makes them well suited to rapid channel adaptation and large-scale integration. Despite the growing research on PRAs, three important issues remain unresolved. First, existing ideal models are often constructed from mathematical considerations and represent continuous pattern functions using orthogonal basis expansions. However, the physical constraints imposed by a finite antenna size  are not explicitly incorporated into these models. This disconnect obscures the physical interpretation of the optimized patterns and, more importantly, may admit patterns that cannot be physically realized, thereby greatly overestimating the performance gains offered by PRAs. Second, ideal and practical PRA models have largely been developed along separate research lines. Consequently, the physical validity of ideal models has not been systematically verified using realizable hardware, while the size- and hardware-induced limitations observed in practical models remain insufficiently explained from a theoretical perspective. Third, practical PRAs often employ numerous discrete RF switches, giving rise to a large set of antenna states. These discrete states are further coupled with the continuous digital beamformer, resulting in substantial design complexity. This difficulty motivates the development of low-complexity optimization methods.

To address these issues, this paper develops a finite-aperture ideal model for PRAs, establishes and validates the relationship between the ideal and practical PRA models, and develops a low-complexity hardware-aware algorithm for PRA-enabled multiuser communication systems. The main contributions are summarized as follows.
\begin{itemize}
	\item We develop a finite-aperture ideal model for PRAs. First, we represent the continuous complex radiation pattern using a Fourier series. By applying the Jacobi-Anger expansion, we then establish an explicit relationship between the antenna aperture and the effective Fourier order and derive aperture-dependent power bounds for the individual Fourier coefficients. The resulting finite-dimensional model provides a physically grounded, implementation-independent capability envelope for PRAs confined to a prescribed aperture.
	
	\item We establish and validate the relationship between the ideal and practical PRA models. Specifically, we design three reconfigurable pixel antennas (RPAs) with different aperture sizes and pixel layouts, and construct multiport network models that map their discrete switch states to radiation patterns. After validating the multiport network models against full-wave electromagnetic simulations, we represent the practically realizable patterns in the same Fourier domain as the ideal model. The results show that the practical patterns exhibit the aperture-dependent effective Fourier orders and satisfy the per-order power constraints in the ideal model, while the specific pixel and switch configurations determine the realizable subset within the ideal capability envelope. These results validate the physical relevance of the ideal model and provide a theoretical interpretation of the practical PRA characteristics.
	
	\item We formulate the weighted sum-rate (WSR) maximization problem for PRA-enabled multiuser communications and develop a coefficient-bound penalty-based weighted minimum mean-square error (WMMSE) algorithm for the resulting mixed-binary problem. By introducing  auxiliary continuous Fourier-coefficient variables, the algorithm decomposes the mixed-binary design into continuous Fourier coefficient optimization and discrete switch-state optimization. For the discrete optimization, we propose a one-bit neighborhood search (OBNS) that evaluates all single-bit neighbors of the current state and exploits the Sherman-Morrison identity to update the inverse of the loaded multiport impedance matrix.
\end{itemize}
 
The rest of this paper is organized as follows. Section~\ref{SysModel} develops the ideal and practical PRA models and validates their relationship. Section~\ref{sec:mb-system} presents the system model and problem formulations. Section~\ref{sec:mb-algorithm} develops the proposed optimization algorithm and analyzes its convergence as  well as computational complexity. Section~\ref{SR} presents the simulation results, and Section~\ref{CON} concludes the paper.

\textit{Notation:} Scalars, vectors, matrices, and sets are denoted by italic letters, boldface lowercase letters, boldface uppercase letters, and calligraphic letters, respectively. The superscripts $(\cdot)^*$, $(\cdot)^{\rm T}$ and $(\cdot)^{\rm H}$ denote the conjugate, transpose and conjugate transpose, respectively. The operators $|\cdot|$, $\|\cdot\|_2$, and $\|\cdot\|_{\rm F}$ denote the absolute value, Euclidean norm, and Frobenius norm, respectively. Moreover, $\mathbb C$, $\mathbb R$, and $\mathbb Z$ denote the sets of complex, real, and integer numbers, respectively. $\mathbb E\{\cdot\}$ denotes statistical expectation and $\mathrm{diag}\{\cdot\}$ constructs a diagonal matrix.

\section{Modeling of PRAs}\label{SysModel}
In this section, we investigate the modeling of PRAs. We first introduce the operating principle of PRAs and discuss their relationship with position reconfiguration, as commonly considered in MA and FAS. Then, we develop a  finite-aperture ideal model to characterize the reconfigurability of PRAs based on a Fourier series and the Jacobi-Anger expansion. Moreover, we present a multiport network model of RPAs as a representative practical realization of PRAs. Finally, we validate the developed ideal model using radiation patterns generated by practical pixel-antenna designs and establish a link between the ideal and practical models.

\begin{figure}[!t]
	\centering
	\includegraphics[width=90mm]{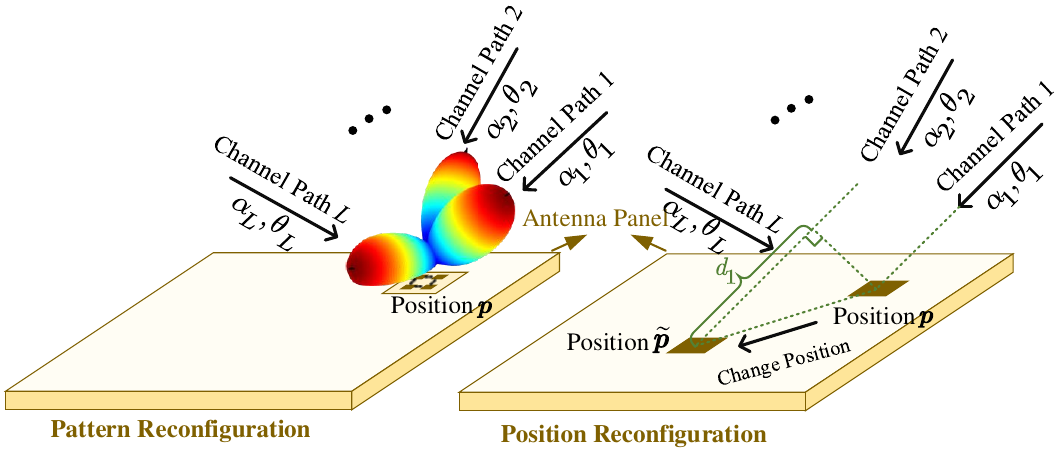}
	\caption{Multipath reweighting through pattern and position reconfiguration.
		\label{fig_image1}}
\end{figure}

\subsection{Introduction to PRAs}
A PRA is an antenna whose radiation pattern can be dynamically altered by controlling tunable components embedded in its radiating structure. In a two-dimensional propagation environment, its complex far-field response can be represented as a continuous function of the spatial angle $\theta \in [-\pi,\pi)$, denoted by $f(\theta)\in\mathbb{C}$.

To demonstrate how the radiation pattern affects multipath reception, we consider a radiation source whose signal reaches a PRA through $L$ propagation paths, as illustrated on the left of Fig.~\ref{fig_image1}. The channel gain and angle of the $l$th path are denoted by $\alpha_l$ and $\theta_l$, respectively. The resulting effective channel can be expressed as
\begin{align}\label{CRPRA}
	h = \sum_{l=1}^{L}f(\theta_l)\alpha_l.
\end{align}
According to \eqref{CRPRA}, the effective channel is a coherent superposition of the multipath components, where $f(\theta_l)$ determines the complex weight applied to the $l$th path. Reconfiguring the radiation pattern therefore changes the relative amplitudes and phases of these components and consequently reshapes the effective channel.

We next compare pattern reconfiguration with the position-reconfiguration mechanism commonly considered in MA and FAS, as illustrated on the right of Fig.~\ref{fig_image1}. Suppose that an antenna is displaced from $\boldsymbol{p}$ to $\widetilde{\boldsymbol{p}}$. Neglecting the variation in large-scale path loss caused by the displacement, the effective channel becomes
\begin{align}\label{CRMF}
	\overline{h}
	= \sum_{l=1}^{L} e^{j2\pi (d_l(\widetilde{\boldsymbol{p}})-d_l({\boldsymbol{p}}))/\lambda}\alpha_l
	= \sum_{l=1}^{L}e^{j\psi_l(\widetilde{\boldsymbol{p}})}\alpha_l,
\end{align}
where $d_l({\boldsymbol{p}})$ denotes the propagation distance associated with the $l$th path at position ${\boldsymbol{p}}$, $\psi_l(\widetilde{\boldsymbol{p}})\triangleq 2\pi (d_l(\widetilde{\boldsymbol{p}})-d_l({\boldsymbol{p}}))/\lambda$, and $\lambda$ denotes the carrier wavelength. \eqref{CRMF} shows that position reconfiguration also changes the coherent superposition of the propagation paths by modifying their complex weights.

From \eqref{CRPRA} and \eqref{CRMF}, both reconfiguration mechanisms modify the coherent superposition of the propagation paths by changing their complex weights. However, the underlying mechanisms and the resulting constraints on these weights are different. Position reconfiguration modifies the path phases through physical displacement, resulting in unit-modulus weights $e^{j\psi_l(\widetilde{\boldsymbol{p}})}$ that are jointly determined by a common position vector. In contrast, pattern reconfiguration directly modifies the angular response $f(\theta_l)$ and can therefore adjust both the amplitudes and phases of the path weights, subject to the set of physically realizable radiation patterns. Consequently, a PRA can reproduce the same effective channel as position reconfiguration if the corresponding position-induced path weights are realizable by its radiation pattern. Moreover, the two mechanisms induce different feasible sets of path weights. Specifically, position reconfiguration is constrained by the unit-modulus weights jointly determined by the antenna position, whereas pattern reconfiguration is constrained by the physically realizable radiation patterns.

Note that the equivalence between  position reconfiguration and pattern reconfiguration has been investigated in previous works~\cite{zhang2026pixel,chen2025remaa,Zhang2025PRAFAS}. These studies primarily establish the equivalence by matching spatial correlation characteristics and focus on emulating physical antenna movement. In contrast, this work interprets the relationship from a path-wise channel-transformation perspective. As discussed above, pattern reconfiguration is not restricted to position-induced unit-modulus path weighting, but has its own feasible set of complex path weights determined by the physically realizable radiation patterns. Therefore, PRAs should not be viewed merely as a means of emulating antenna movement. Their distinct reconfigurability may provide additional degrees of freedom for manipulating multipath superposition. Motivated by this observation, we next characterize the intrinsic reconfigurability of PRAs through ideal and practical radiation-pattern models.

\subsection{Ideal Model for PRAs}~\label{IdealModel}
We next develop an implementation-independent ideal model to characterize the intrinsic pattern reconfigurability of a PRA. The radiation pattern $f(\theta)$ is a continuous complex-valued function of the spatial angle, and directly treating $f(\theta)$ as a design variable leads to an infinite-dimensional formulation. To obtain a tractable parametric representation, we expand it using a Fourier series as
\begin{align}\label{FourierTransform}
	f(\theta)=\sum_{n=-\infty}^{\infty}\gamma_n e^{jn\theta},
\end{align}
where $\gamma_n\in\mathbb C$ denotes the $n$th Fourier coefficient.

\subsubsection{Order Constraint}\label{OC}
Although \eqref{FourierTransform} converts the functional design into the design of Fourier coefficients, it still contains infinitely many variables. To obtain a finite-dimensional representation, we next exploit the physical constraint imposed by the finite antenna size.

\begin{proposition}
	Let $R_{\rm A}$ denote the radius of the smallest circle enclosing the antenna aperture.
	The Fourier coefficients of $f(\theta)$ admit a rapidly decaying upper envelope for
	$|n|>e\pi R_{\rm A}/\lambda$. Accordingly, an effective Fourier truncation order is
	given by
	\begin{align}
		N=\left\lceil \frac{e\pi R_{\rm A}}{\lambda}\right\rceil .
	\end{align}
\end{proposition}
\textit{Proof: See Appendix A.}

Proposition 1 indicates that the high-order Fourier coefficients are increasingly
suppressed beyond an aperture-dependent effective order. Accordingly, retaining
the $2N+1$ Fourier orders with $|n|\leq N$, where $N$ is determined by the
electrical size of the antenna aperture, provides the finite-dimensional approximation
\begin{align} \label{FourierTransform2}
	f(\theta) \approx \sum_{n=-N}^{N}\gamma_n e^{jn\theta}.
\end{align}


\subsubsection{Total Power Constraint}\label{TPC}
To eliminate arbitrary scaling of the radiation pattern, we normalize its angular power to unity. Using \eqref{FourierTransform2}, we have
\begin{align}\label{PatternPower}
	\frac{1}{2\pi} \int_{-\pi}^{\pi} |f(\theta)|^2 {\rm d}\theta&\approx\frac{1}{2\pi} \int_{-\pi}^{\pi} \bigg| \sum_{n=-N}^{N}\gamma_n e^{jn\theta}\bigg|^2 {\rm d}\theta\nonumber\\
	&\overset{\rm (a)}{=}\sum_{n=-N}^{N} |\gamma_n|^2,
\end{align}
where (a) follows from the orthogonality of the Fourier basis functions. Therefore, the normalized Fourier coefficients satisfy
\begin{align}\label{CoefficientPowerConstraint}
	\sum_{n=-N}^{N} |\gamma_n|^2 =1. 
\end{align}

\subsubsection{Per-Order Power Constraint}\label{PCE}
The total power constraint alone does not capture the order-dependent limitations imposed by the finite antenna size. In particular, sufficiently high-order angular modes are strongly suppressed for an antenna with a finite electrical size~\cite{Bucci1987Spatial}. Therefore, an additional per-order power constraint is imposed on the Fourier coefficients. As derived in Appendix~B, the power of the $n$th-order coefficient satisfies
\begin{align}\label{TB}
	|\gamma_n|^2\leq\min\{1,C B_n\},
\end{align}
where $B_n$ is an aperture-dependent factor and $C$ characterizes the admissible aperture-current intensity. For convenience, define
\begin{align}
	\rho_n\triangleq\sqrt{\min\{1,C B_n\}}
\end{align}
as the amplitude bound of the $n$th Fourier coefficient.

Combining the order, total-power, and per-order power constraints, the proposed ideal PRA model is given by
\begin{align}
	&f(\theta)\approx
	\sum_{n=-N}^{N}
	\gamma_n e^{jn\theta},\nonumber\\
	&\sum_{n=-N}^{N}
	|\gamma_n|^2=1,\quad|\gamma_n|\leq\rho_n,\quad
	n=-N,\ldots,N.
\end{align}

Note that ideal models for PRAs have also been considered in existing works~\cite{Ying2025Reconfigurable,Zheng2026TriHybrid}. The developed model extends these formulations in three aspects. First, it adopts a complex-valued Fourier representation that explicitly captures both the amplitude and phase of the angular response. Such joint amplitude-phase characterization is consistent with practical PRA designs~\cite{WangPhaseAmp2024}, and is particularly suitable for manipulating the coherent superposition of multipath components. Second, the truncation order is explicitly related to the finite electrical size of the antenna aperture. Third, an order-dependent power bound is imposed on each angular mode to capture the aperture-induced suppression of high-order modes. Taken together, the truncation-order relation and the per-order power bounds specify two fundamental physical conditions that an ideal PRA model should satisfy.

In summary, the proposed ideal model provides an antenna aperture-dependent theoretical capability envelope on the radiation patterns achievable by a PRA confined to a prescribed finite aperture. 


\subsection{Practical Model  for PRAs}~\label{PracticalModel}
The radiation pattern is governed by the surface-current distribution over the antenna radiating structure. Therefore, reconfiguring the radiation pattern requires the surface current to be redistributed. An RPA provides a practical architecture for this purpose. As shown on the left side of Fig.~\ref{fig_image2}, an RPA consists of metallic pixels interconnected by RF switches, such as PIN diodes. Switching each PIN diode between the ``on'' and ``off'' states changes the current paths over the antenna surface and reconfigures the resulting radiation pattern.

\begin{figure}[!t]
	\centering
	\includegraphics[width=90mm]{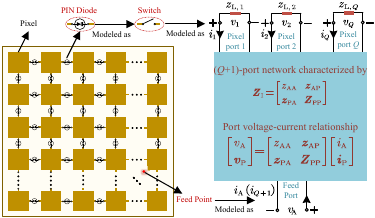}
	\caption{Equivalent multiport network model of an RPA.
		\label{fig_image2}}
\end{figure}

A multiport network model is widely used to characterize the effect of the PIN diode states on the radiation pattern~\cite{Shen2026AntennaCoding}. In this model, each feed point and each PIN-diode connection is represented as a port, while each PIN diode is represented by a switchable load. Let $Q$ denote the number of PIN diodes. The resulting network contains $Q+1$ ports, comprising one feed port and $Q$ pixel ports. The current and voltage of the $q$th pixel port, $q=1,2,\cdots,Q$, are denoted by $i_q$ and $v_q$, respectively. For notational convenience, define $\boldsymbol{i}_{\rm P} \triangleq [i_1,i_2,\cdots,i_Q]^{\rm T}$ and $\boldsymbol{v}_{\rm P} \triangleq [v_1,v_2,\cdots,v_Q]^{\rm T}$. The feed-port current and voltage are denoted by $i_{\rm A}$ and $v_{\rm A}$, respectively. The port voltages $[v_{\rm A},\boldsymbol{v}_{\rm P}^{\rm T}]^{\rm T}$ and currents $[i_{\rm A},\boldsymbol{i}_{\rm P}^{\rm T}]^{\rm T}$ are related through the network impedance matrix
\begin{align}
	\boldsymbol{Z}_{\rm I} = \left[ \begin{matrix}
		z_{\mathrm{A}\mathrm{A}}&		\boldsymbol{z}_\mathrm{AP}\\
		\boldsymbol{z}_\mathrm{PA}&		\boldsymbol{Z}_\mathrm{PP}\\
	\end{matrix} \right],
\end{align}
where $z_{\mathrm{A}\mathrm{A}}$, $\boldsymbol{z}_\mathrm{AP}$, $\boldsymbol{z}_\mathrm{PA}$, and $\boldsymbol{Z}_\mathrm{PP}$ denote the feed-port self-impedance, the feed-pixel mutual-impedance row and column vectors, and the pixel-port impedance matrix, respectively. The port voltage-current relationship is
\begin{align}\label{portrelationship}
	\left[\begin{matrix}
		v_\mathrm{A}\\
		\boldsymbol{v}_\mathrm{P}\\
	\end{matrix} \right] = \left[\begin{matrix}
		z_{\mathrm{A}\mathrm{A}}&		\boldsymbol{z}_\mathrm{AP}\\
		\boldsymbol{z}_\mathrm{PA}&		\boldsymbol{Z}_\mathrm{PP}\\
	\end{matrix} \right]\left[\begin{matrix}
		i_\mathrm{A}\\
		\boldsymbol{i}_\mathrm{P}\\
	\end{matrix} \right].
\end{align}

We next model the PIN diodes. Each PIN diode operates in either the ``on'' or ``off'' state, and the collective switch state is represented by the binary vector $\boldsymbol{t}\in\mathbb{Z}^{Q}$. Specifically, $[\boldsymbol{t}]_q=0$ and $[\boldsymbol{t}]_q=1$ represent the ``on'' and ``off'' states of the $q$th PIN diode, respectively. Under the ideal-switch approximation, the impedance of a PIN diode is nearly zero in the ``on'' state and approaches infinity in the ``off'' state. Accordingly, the PIN diodes are represented by the load-impedance matrix
\begin{align}
	\boldsymbol{Z}_{\rm L}(\boldsymbol{t}) = \mathrm{diag}\{z_{{\rm L},1},z_{{\rm L},2},\cdots,z_{{\rm L},Q}\},
\end{align}
where the $q$th diagonal entry of $\boldsymbol{Z}_{\rm L}(\boldsymbol{t})$ is
\begin{align}
	z_{{\rm L},q} = \begin{cases}
		0	,~~~~~[\boldsymbol{t}]_q = 0,\\
		\infty	,~~~~[\boldsymbol{t}]_q = 1.\\
	\end{cases}
\end{align}
In numerical evaluations, the infinite off-state impedance is approximated by a sufficiently large value, e.g., $10^{6}$. The load impedance, pixel-port current, and pixel-port voltage satisfy
\begin{align}\label{vP1}
	\boldsymbol{v}_\mathrm{P} = -\boldsymbol{Z}_{\rm L}(\boldsymbol{t})\boldsymbol{i}_\mathrm{P}.
\end{align}
From \eqref{portrelationship}, the pixel-port voltage also satisfies
\begin{align}\label{vP2}
	\boldsymbol{v}_\mathrm{P} = \boldsymbol{z}_\mathrm{PA}i_\mathrm{A} + \boldsymbol{Z}_\mathrm{PP}\boldsymbol{i}_\mathrm{P}.
\end{align}
Combining \eqref{vP1} and \eqref{vP2} yields
\begin{align}\label{Ip}
	\boldsymbol{i}_\mathrm{P} = -(\boldsymbol{Z}_{\rm L}(\boldsymbol{t})+\boldsymbol{Z}_\mathrm{PP})^{-1}\boldsymbol{z}_\mathrm{PA}i_\mathrm{A}.
\end{align}

By electromagnetic linearity, the radiation pattern of the RPA is the superposition of the open-circuit embedded patterns of the feed and pixel ports, weighted by their respective port currents. Hence, the resulting radiation pattern is
\begin{align}\label{matcal}
	g(\theta) = \sum_{q=1}^{Q+1} i_{q} b_q(\theta),
\end{align}
where $b_q(\theta)$, $q=1,2,\cdots,Q$, denotes the open-circuit embedded radiation pattern associated with the $q$th pixel port. We stipulate that $i_{Q+1} = i_{\rm A}$ for ease of notation and $b_{Q+1}(\theta)$ is the open-circuit embedded radiation pattern associated with the feed port. Note that the impedance parameters and embedded radiation patterns required by this practical model can be obtained from full-wave electromagnetic simulations.

\begin{figure}[!t]
	\centering
	\includegraphics[width=85mm]{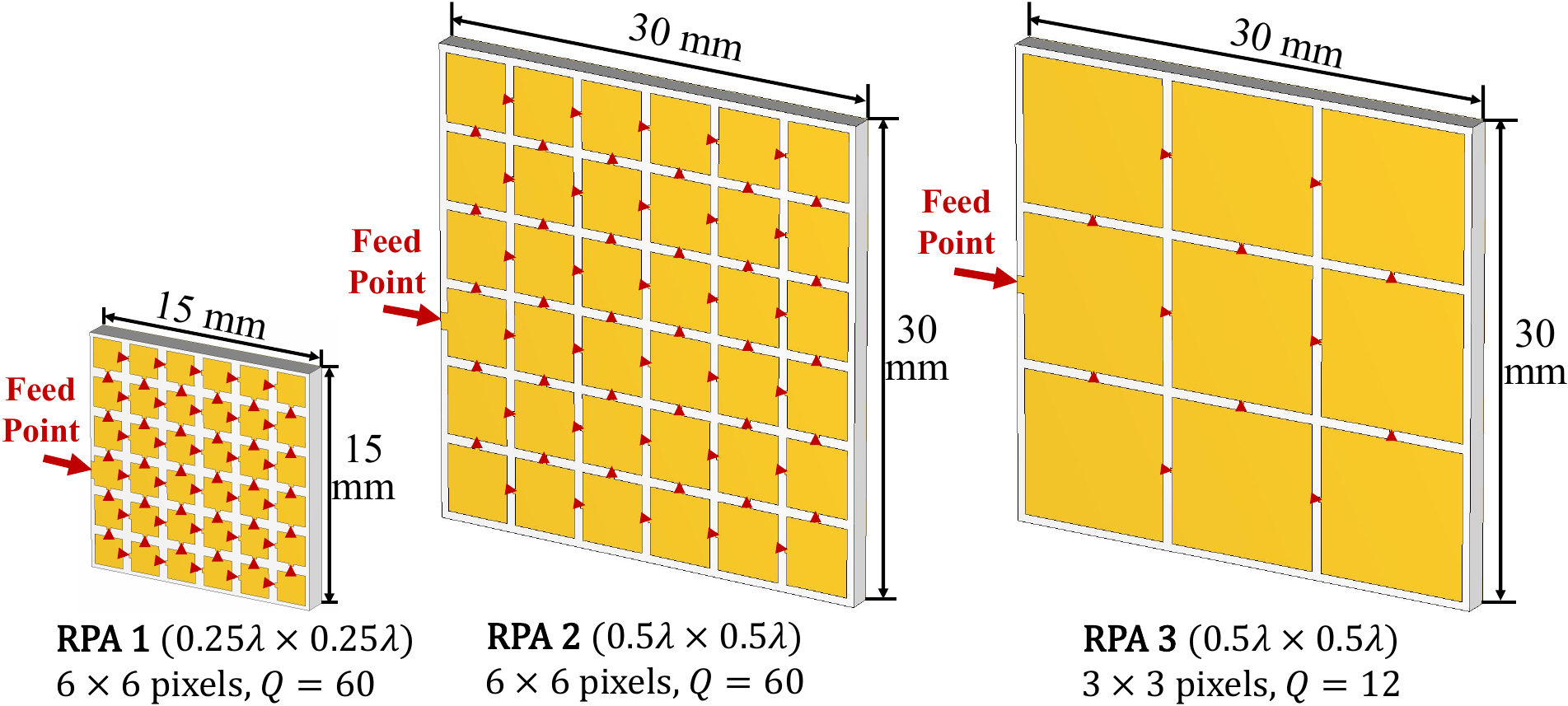}
	\caption{Three RPA designs at 5 GHz.
		\label{fig_image3}}
				\vspace{-0.2cm}
\end{figure}

\begin{figure}[!t]
	\centering
	\includegraphics[width=85mm]{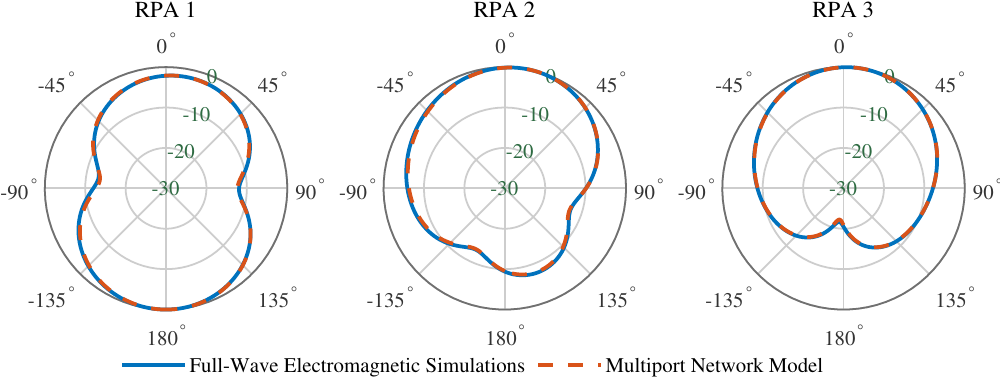}
	\caption{Normalized radiation patterns obtained by full-wave electromagnetic simulations and the multiport network model. 
	\label{fig_image4}}
		\vspace{-0.5cm}
\end{figure}

\subsection{Consistency Between the Ideal and Practical Models}
\label{ModelVerification}

The ideal and practical PRA models characterize pattern reconfigurability from complementary perspectives. The ideal model in Section~\ref{IdealModel} characterizes the implementation-independent capability envelope imposed by a prescribed antenna size, whereas the practical model in Section~\ref{PracticalModel} describes the radiation patterns realizable by a specific RPA architecture. Their consistency can therefore be examined by comparing the Fourier-domain characteristics of practically realizable radiation patterns with the constraints determined by the ideal model. Specifically, practical RPA designs are used to verify the aperture-dependent Fourier order and per-order coefficient bounds of the ideal model, while the ideal model provides a theoretical reference for interpreting the Fourier-domain behavior and realizable coefficient regions of practical RPAs.

To examine this consistency, we consider three RPA designs, denoted by RPA~1, RPA~2, and RPA~3, as illustrated in Fig.~\ref{fig_image3}. RPA~1 employs a $6\times6$ pixel layout within a $0.25\lambda\times0.25\lambda$ aperture, RPA~2 employs a $6\times6$ pixel layout within a $0.5\lambda\times0.5\lambda$ aperture, and RPA~3 employs a $3\times3$ pixel layout within the same $0.5\lambda\times0.5\lambda$ aperture. These configurations allow the effects of the aperture size and hardware configuration to be distinguished. Specifically, RPA~1 and RPA~2 have the same pixel layout but different aperture sizes, whereas RPA~2 and RPA~3 have the same aperture size but different pixel layouts. Consistent with the two-dimensional angular model adopted in this work, the following validation is performed using the complex far-field responses over a fixed angular cut. 

We first verify the practical multiport network model by comparing the radiation patterns calculated using \eqref{matcal} with those obtained from full-wave electromagnetic simulations. As shown in Fig.~\ref{fig_image4}, the radiation patterns obtained using the two approaches are in close agreement. This agreement confirms that the multiport network model accurately captures the dependence of the practical radiation pattern on the switch configuration and can therefore be used to investigate the Fourier-domain characteristics of the three RPAs.

\begin{figure}[!t]
	\centering
	\includegraphics[width=75mm]{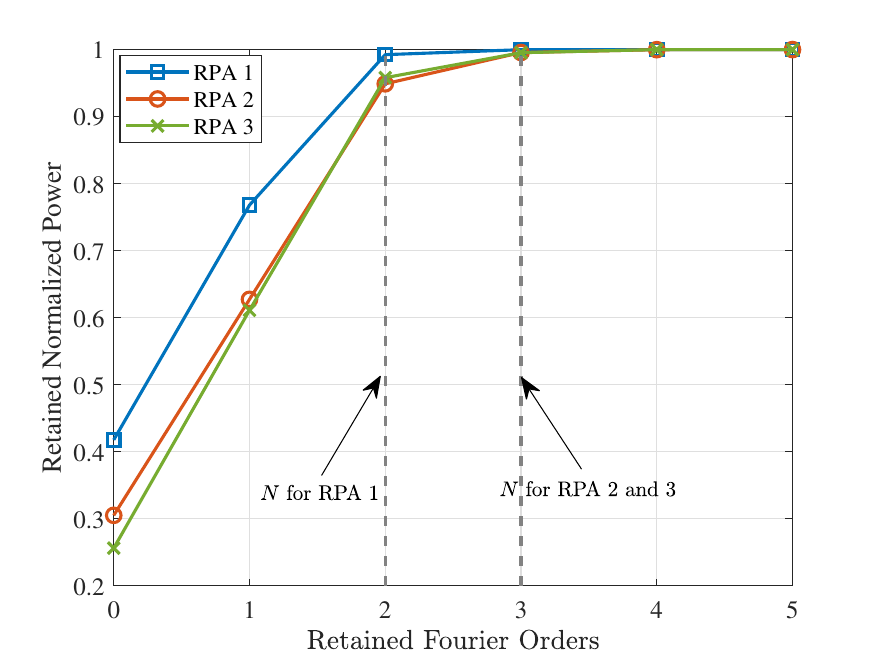}
	\caption{Average retained normalized power versus the retained Fourier order. 
		\label{fig_image5}}
				\vspace{-0.5cm}
\end{figure}

We next examine the order constraint developed in Section~\ref{OC}. Fig.~\ref{fig_image5} shows the average normalized pattern power retained when the practical radiation patterns are reconstructed using progressively increasing Fourier orders. According to Proposition~1, using the projected aperture radius in the horizontal plane, the predicted truncation order $N$ is $2$ for RPA~1 and $3$ for RPA~2 and RPA~3. When the retained Fourier order reaches the corresponding $N$, approximately $99.2\%$ of the pattern power is retained for RPA~1 and $99.5\%$ for RPA~2 and RPA~3. These results validate the aperture-dependent order constraint predicted by the ideal model. In particular, RPA~2 and RPA~3 have the same physical aperture and consequently exhibit the same predicted effective Fourier orders despite their different pixel layouts. If a higher retained-power ratio is required in practice, an additional truncation margin $U_{\rm r}\in\mathbb Z_{+}$ can be introduced to capture the residual high-order components.

\begin{figure}[!t]
	\centering
	\includegraphics[width=75mm]{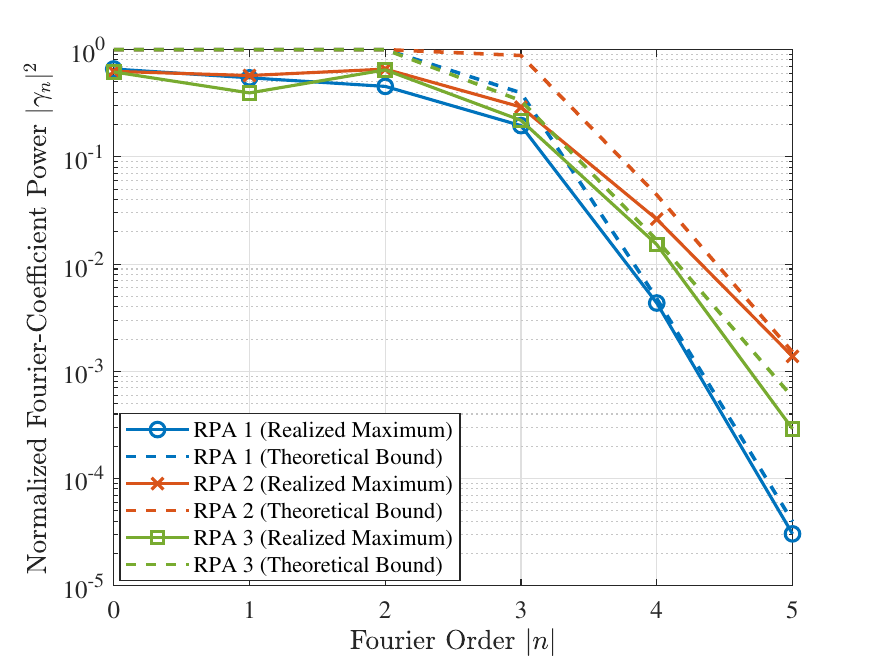}
	\caption{Normalized Fourier-coefficient power versus Fourier order.
		\label{fig_image6}}
		\vspace{-0.5cm}
\end{figure}

We further examine the per-order power constraint in Section~\ref{PCE}. Fig.~\ref{fig_image6} compares the maximum normalized Fourier-coefficient power realized by the three RPAs with the corresponding theoretical bounds in \eqref{TB}. For all three designs, the realized Fourier-coefficient powers remain below the theoretical bounds. Moreover, the coefficient powers become negligible at sufficiently high Fourier orders, consistent with the aperture-induced limitation characterized by the ideal model. These observations further validate the per-order Fourier-coefficient power constraints derived for the ideal model. Meanwhile, a noticeable gap remains between the realized maxima and the theoretical bounds. This gap is expected because the theoretical bounds characterize the implementation-independent capability envelope associated with a prescribed aperture, whereas a practical RPA is additionally constrained by its pixel geometry, switch topology, and realizable current distributions.

We finally examine the realizable Fourier-coefficient regions of the three RPAs in Fig.~\ref{fig_image7}. Each sample is obtained from a switch configuration by calculating the corresponding radiation pattern and extracting its Fourier coefficients. The samples of RPA~1 and RPA~2 form denser and more diverse realizable sets than those of RPA~3, showing that the hardware configuration affects the range of achievable Fourier coefficients. In particular, although RPA~2 and RPA~3 have the same aperture size and thus the same effective Fourier order, their different pixel layouts lead to substantially different realizable coefficient regions.

The above results demonstrate the consistency between the ideal and practical PRA models. The practical RPA designs exhibit the aperture-dependent effective Fourier order and per-order Fourier-coefficient power limitations predicted by the ideal model. In turn, the ideal model provides an implementation-independent theoretical reference for interpreting the Fourier-domain characteristics of practical RPAs: the antenna aperture mainly determines the effective Fourier order and the corresponding capability envelope, whereas the pixel layout and switch configuration determine the realizable coefficient region within this envelope. Therefore, the ideal and practical models provide complementary descriptions of PRA reconfigurability from theoretical and implementation perspectives, respectively.

\section{System Model and Problem Formulation}
\label{sec:mb-system}

Based on the PRA models developed in Section~\ref{SysModel}, we establish a Fourier-domain downlink model for PRA-enabled multiuser MIMO and formulate the WSR maximization problems under the ideal and practical PRA models.

\begin{figure}[!t]
	\centering
	\includegraphics[width=90mm]{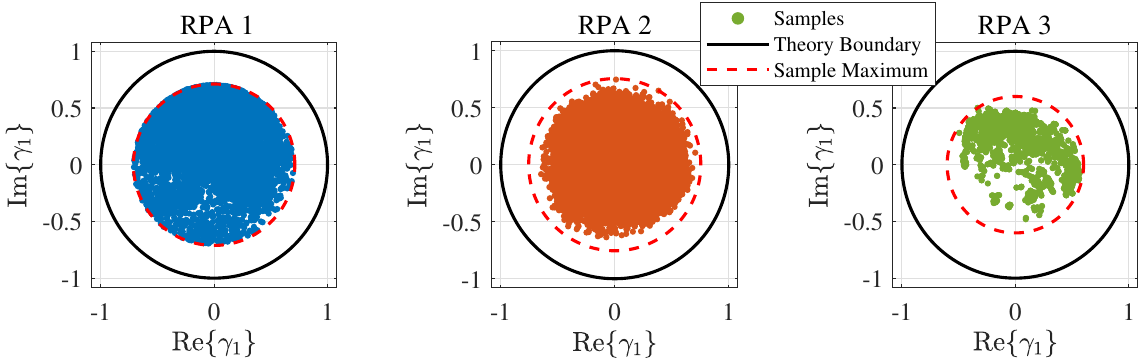}
	\caption{Realized  first-order  Fourier-coefficient samples of the three RPAs.
		\label{fig_image7}}
				\vspace{-0.5cm}
\end{figure}
 
\subsection{System Model}
\label{subsec:mb-downlink}

We consider a narrowband downlink system in which a base station (BS) equipped with $M$ PRAs simultaneously serves $K$ single-antenna users. Each PRA is driven by a dedicated RF chain, and the BS employs the fully digital beamformer $\boldsymbol W=[\boldsymbol w_1,\ldots,\boldsymbol w_K]\in\mathbb C^{M\times K}$. Let $\boldsymbol s=[s_1,\ldots,s_K]^{\mathrm T}\in\mathbb C^K$ denote the data
symbol vector, where $\mathbb E\{\boldsymbol s\boldsymbol s^{\mathrm H}\}=\boldsymbol I_K$. The transmit-power constraint is
\begin{align}
	\|\boldsymbol W\|_{\mathrm F}^2\leq P_0.
	\label{eq:mb-transmit-power}
\end{align}

Under the ideal PRA model in Section~\ref{IdealModel}, the radiation pattern of
the $m$th PRA is represented as
\begin{align}
	f_m(\theta)=\boldsymbol f(\theta)^{\mathrm T}\boldsymbol\gamma_m,
	\label{IModel}
\end{align}
where $\boldsymbol f(\theta)=[e^{-jN\theta},\ldots,1,\ldots,e^{jN\theta}]^{\mathrm T} \in\mathbb C^T$ contains the retained angular basis functions and $\boldsymbol\gamma_m=[\gamma_{m,-N},\ldots,\gamma_{m,N}]^{\mathrm T} \in\mathbb C^T$ contains the corresponding Fourier coefficients, with $T\triangleq2N+1$.

For the practical PRA model, let $\boldsymbol t_m\in\{0,1\}^{Q}$ denote the PIN diode state vector of the $m$th PRA. According to the multiport network model, its radiation pattern is
\begin{align}
	g_m(\theta)
	=\boldsymbol b(\theta)^{\mathrm T}\boldsymbol i_m(\boldsymbol t_m),
	\label{PModel}
\end{align}
where $\boldsymbol b(\theta)=[b_1(\theta),\ldots,b_{Q+1}(\theta)]^{\mathrm T} \in\mathbb C^{Q+1}$ collects the embedded port patterns and
$\boldsymbol i_m(\boldsymbol t_m)=[i_{m,1}(\boldsymbol t_m),\ldots,i_{m,Q+1}(\boldsymbol t_m)]^{\mathrm T} \in\mathbb C^{Q+1}$ collects the corresponding port currents. Here, $Q$ is the number of PIN diodes, and the $(Q+1)$ ports comprise $Q$ pixel ports and one feed port.

The consistency analysis in Section~\ref{ModelVerification} shows that practical patterns can be represented in the same Fourier domain as ideal patterns. Accordingly, the truncated representation of \eqref{PModel} is
\begin{align}\label{PModel2}
	g_m(\theta) \approx\boldsymbol f(\theta)^{\mathrm T} \boldsymbol\xi_m(\boldsymbol t_m),
\end{align}
where
$\boldsymbol\xi_m(\boldsymbol t_m) =[\xi_{m,-N}(\boldsymbol t_m),\ldots,\xi_{m,N}(\boldsymbol t_m)]^{\mathrm T} \in\mathbb C^T$ contains the retained Fourier coefficients of the practical pattern. Projecting \eqref{PModel} onto the Fourier basis gives
\begin{align}
	\xi_{m,n}(\boldsymbol t_m)
	&=\sum_{q=1}^{Q+1}c_{n,q}i_{m,q}(\boldsymbol t_m),~ n=-N,\ldots,N,
	\label{eq:practical-modal-coefficient}\\
	c_{n,q}
	&\triangleq\frac{1}{2\pi}\int_{-\pi}^{\pi}
	b_q(\theta)e^{-jn\theta}\,\mathrm d\theta.
	\label{eq:embedded-pattern-coefficient}
\end{align}
Define $\boldsymbol C\in\mathbb C^{T\times(Q+1)}$ by $[\boldsymbol C]_{n+N+1,q}=c_{n,q}$. Then, the practical multiport model and the common Fourier-domain model are linked by the mapping
\begin{align}
	\boldsymbol\xi_m(\boldsymbol t_m)
	=\boldsymbol C\boldsymbol i_m(\boldsymbol t_m).
	\label{eq:hardware-to-modal-map}
\end{align}

In this paper, we focus on the channel-shaping gain provided by pattern reconfiguration. To keep the same total-pattern-power scale as the ideal models, we normalize $\boldsymbol{\xi}_m(\boldsymbol{t}_m)$ as 
\begin{align}
	\boldsymbol{\zeta}_m(\boldsymbol{t}_m) = \boldsymbol{\xi}_m(\boldsymbol{t}_m)/\|\boldsymbol{\xi}_m(\boldsymbol{t}_m)\|_2.
\end{align}

From \eqref{IModel} and \eqref{PModel2}, the radiation patterns of PRAs are composed of the channel angle-dependent responses and the electromagnetic design parameters. By incorporating the channel-angle-dependent responses into the Saleh--Valenzuela channel model, we express the Fourier-domain channel between the $m$th antenna and the $k$th user as
\begin{align}\label{modalchannel}
	\boldsymbol{h}_{m,k} = \sum_{l=1}^{L_k}\alpha_{k,l}^*a_m^*(\theta_{k,l})\boldsymbol f^*(\theta_{k,l}),
\end{align}
where $L_k$ is the number of channel paths, and $\alpha_{k,l}$ and $\theta_{k,l}$ are the channel gain and  angle of the $l$th path, respectively. $a_m(\theta_{k,l})$ is the array response of the $m$th antenna. Denote the position of the $m$th antenna by $x_m$. Then, $a_m(\theta_{k,l}) = e^{j2\pi x_m\sin\theta_{k,l}/\lambda}$. Stacking $\boldsymbol{h}_{m,k}$ for $M$ PRAs yields
\begin{align}
	\overline{\boldsymbol{h}}_k=\left[\boldsymbol{h}_{1,k}^{\rm T},\ldots,\boldsymbol{h}_{M,k}^{\rm T}\right]^{\rm T}\in\mathbb{C}^{MT}.
\end{align}

Stacking the Fourier coefficients for the $M$ antennas, we obtain the electromagnetic beamformers for the ideal and practical models as
\begin{align}\label{EMBeamformer}
\boldsymbol{\Gamma} &= \mathrm{diag}\{\boldsymbol{\gamma}_1,\cdots,\boldsymbol{\gamma}_M\},\nonumber\\
\boldsymbol{\varLambda}(\boldsymbol{\overline{t}}) &= \mathrm{diag}\{\boldsymbol{\zeta}_1(\boldsymbol{t}_1),\cdots,\boldsymbol{\zeta}_M(\boldsymbol{t}_M)\},
\end{align}
 where $\boldsymbol{\overline{t}}\triangleq[\boldsymbol{t}_1^{\rm T},\boldsymbol{t}_2^{\rm T},\cdots,\boldsymbol{t}_M^{\rm T}]^{\rm T}$.

With the Fourier-domain channel in \eqref{modalchannel} and the ideal electromagnetic beamformer in \eqref{EMBeamformer}, the received signal for the $k$th user can be expressed as 
\begin{align}\label{eq:received}
	y_k=\boldsymbol{\overline{h}}_k^{\rm H}\boldsymbol{\Gamma}\boldsymbol{w}_k s_k+\sum_{i\neq k}\boldsymbol{\overline{h}}_k^{\rm H}
   \boldsymbol{\Gamma}\boldsymbol{w}_i s_i+n_k,
\end{align}
where $n_k\sim\mathcal{CN}(0,\sigma^2)$ denotes the additive white Gaussian noise with power $\sigma^2$. By replacing $\boldsymbol{\Gamma}$ with $\boldsymbol{\varLambda}(\boldsymbol{\overline{t}})$, we can obtain the received signal under the practical electromagnetic beamformer.

\subsection{Problem Formulation} \label{subsec:mb-common-modal}
\subsubsection{Ideal PRA Model} From \eqref{eq:received}, the signal-to-interference-plus-noise ratio (SINR) of the $k$th user under the ideal PRA model is expressed as
\begin{align}\label{SINR}
\varUpsilon_k(\boldsymbol{W},\boldsymbol{\Gamma}) =\frac{|\boldsymbol{\overline{h}}_k^{\rm H}\boldsymbol{\Gamma}\boldsymbol{w}_k|^2}
	 {\sum_{i\neq k}|\boldsymbol{\overline{h}}_k^{\rm H}\boldsymbol{\Gamma}\boldsymbol{w}_i|^2+\sigma^2}.
\end{align}
Based on \eqref{SINR}, the WSR of $K$ users can be expressed as 
\begin{align}\label{sumrate}
	S(\boldsymbol{W},\boldsymbol{\Gamma}) = \sum_{k=1}^{K}\chi_k\log_2(1+\varUpsilon_k(\boldsymbol{W},\boldsymbol{\Gamma})),
\end{align}
where $\chi_k$ is the weighting coefficient of the $k$th user.

According to the ideal PRA model in Section~\ref{IdealModel}, the feasible set of $\boldsymbol{\gamma}_m$ is
\begin{align}\label{feasibleset}
	\mathcal{T}\!=\!\left\{\boldsymbol\gamma\in\mathbb C^T\big|\|\boldsymbol\gamma\|_2^2\!=\!1, |\gamma_n|\leq  \rho_{n},~n\!=\!-N,\ldots,N\right\}.
\end{align}

We aim to maximize the WSR of $K$ users subject to the feasible constraint of $\boldsymbol{\gamma}_m$ in \eqref{feasibleset} and the power constraint of $\boldsymbol{W}$ in \eqref{eq:mb-transmit-power}, which can be formulated as 
\begin{subequations}\label{prob:mb-ideal}
	\begin{align}
		(\mathcal P_{\rm I})~~\max_{\boldsymbol W,\boldsymbol{\gamma}_m}\quad&S(\boldsymbol{W},\boldsymbol{\Gamma})\label{prob:mb-ideal-a}\\
		\mathrm{s.t.~}\quad~&\|\boldsymbol W\|_{\mathrm F}^2\leq P_0,\label{prob:mb-ideal-b}\\
		&~\boldsymbol\gamma_m\in\mathcal{T},\quad m=1,\ldots,M.\label{prob:mb-ideal-c}
	\end{align}
\end{subequations}
\subsubsection{Practical PRA Model} By replacing $\boldsymbol{\Gamma}$ with $\boldsymbol{\varLambda}(\boldsymbol{\overline{t}})$, we obtain the SINR of the $k$th user and the WSR of the $K$ users under the practical PRA model as $\varUpsilon_k(\boldsymbol{W},\boldsymbol{\varLambda}(\boldsymbol{\overline{t}}))$ and $S(\boldsymbol{W},\boldsymbol{\varLambda}(\boldsymbol{\overline{t}}))$, respectively. Unlike the ideal model, which uses continuous Fourier coefficients, the practical model determines each radiation pattern through the PIN diode state vector $\boldsymbol{t}_m$, where $[\boldsymbol{t}_m]_q\in\{0,1\}$. Replacing the corresponding parts in $(\mathcal P_{\rm I})$, we can obtain the practical-model formulation as
\begin{subequations}\label{mbpra}
	\begin{align}
		(\mathcal P_{\rm P}-1)~~\max_{\boldsymbol W,\boldsymbol{t}_m}\quad&S(\boldsymbol{W},\boldsymbol{\varLambda}(\boldsymbol{\overline{t}}))\label{mbpraa}\\
		\mathrm{s.t.}\quad~&\|\boldsymbol W\|_{\mathrm F}^2\leq P_0,\label{mbprab}\\
		&[\boldsymbol{t}_m]_q\in\{0,1\}, \quad q=1,\ldots,Q.\label{mbprac}
	\end{align}
\end{subequations}

In the next section, we develop a coefficient-bound penalty-based WMMSE (CBP-WMMSE) algorithm for $(\mathcal P_{\rm P}-1)$. The ideal problem $(\mathcal P_{\rm I})$ is the continuous counterpart of $(\mathcal P_{\rm P}-1)$ with the discrete hardware mapping removed and can therefore be solved using the continuous optimization component of the proposed algorithm.

\section{CBP-WMMSE Algorithm}
\label{sec:mb-algorithm}

In this section, we develop a CBP-WMMSE algorithm for solving $(\mathcal P_{\rm P}-1)$. By introducing  auxiliary continuous Fourier-coefficient variables, the algorithm decomposes the mixed-binary design into continuous Fourier coefficient optimization and discrete switch-state optimization. For the discrete optimization, we propose an OBNS that evaluates all single-bit neighbors of the current state and exploits the Sherman-Morrison identity to update the inverse of the loaded multiport impedance matrix.

\subsection{Problem Reformulation}
\label{subsec:mb-rwmmse}
In $(\mathcal{P}_{\rm P}-1)$, the binary variables of PIN diodes are involved in the WSR objective. Therefore, the optimization involves the coupling of discrete and continuous variables. To relax this coupling, we introduce  auxiliary continuous Fourier-coefficient variables
\begin{align}\label{etaCons}
\boldsymbol{\eta}_m = \boldsymbol{\zeta}_m(\boldsymbol{t}_m)\in\mathbb{C}^{T}.
\end{align}
$\boldsymbol{\eta}_m\triangleq [\eta_{m,-N},\cdots,\eta_{m,N}]^{\rm T}$ collects the Fourier coefficients of the $m$th PRA. Since $\boldsymbol\eta_m$ is a continuous surrogate for a practical Fourier-coefficient vector, its component amplitudes are bounded by those supported by the hardware. Let $\overline\rho_n$ denote the maximum achievable amplitude of order $n$, which can be obtained through offline characterization of the RPA, i.e., $\overline\rho_n\triangleq\max_{\boldsymbol t_m\in\{0,1\}^{Q}}|\zeta_{m,n}(\boldsymbol t_m)|$. The feasible set of $\boldsymbol\eta_m$ is then
\begin{align}
	\overline{\mathcal T}\!\triangleq\!\left\{\boldsymbol\eta\in\mathbb C^T\middle|\|\boldsymbol\eta\|_2^2\!=\!1,~|\eta_n|\!\leq\!\overline\rho_n,~n\!=\!-N,\ldots,N\right\},
\end{align}
where $\eta_n\triangleq [\boldsymbol\eta]_{n+N+1}$. Define $\boldsymbol{\varXi}\triangleq\mathrm{diag}\{\boldsymbol{\eta}_1,\boldsymbol{\eta}_2,\cdots,\boldsymbol{\eta}_M\}$.  Then, $(\mathcal P_{\rm P}-1)$ can be equivalently written as
\begin{subequations}\label{mbpra2}
	\begin{align}
		(\mathcal P_{\rm P}-2)~~\max_{\boldsymbol W,\boldsymbol{\varXi},\boldsymbol{t}_m}\quad&S(\boldsymbol{W},\boldsymbol{\varXi})\label{mbpraa2}\\
		\mathrm{s.t.~}\quad~~&\eqref{mbprab},~\eqref{mbprac},~\eqref{etaCons},~\boldsymbol{\eta}_m\in\overline{\mathcal T}.
	\end{align}
\end{subequations}
 
According to the full-power property established in~\cite{TSP23ZXT}, any nontrivial stationary solution uses the entire power budget. Then, we can discard the power constraint  and convert $(\mathcal P_{\rm P}-2)$ to
\begin{subequations}\label{mbpra3}
	\begin{align}
		(\mathcal P_{\rm P}-3)~~\max_{\boldsymbol W,\boldsymbol{\varXi},\boldsymbol{t}_m}\quad&\overline{S}(\boldsymbol{W},\boldsymbol{\varXi})\label{mbpraa3}\\
		\mathrm{s.t.~}\quad~~&~\eqref{mbprac},~\eqref{etaCons},~\boldsymbol{\eta}_m\in\overline{\mathcal T},
	\end{align}
\end{subequations}
where 
\begin{align}\label{SINR2}
	&\overline{S}(\boldsymbol{W},\boldsymbol{\varXi}) = \sum_{k=1}^{K}\chi_k\log_2(1+\overline{\varUpsilon}_k(\boldsymbol{W},\boldsymbol{\varXi})),\nonumber\\
	&\overline{\varUpsilon}_k(\boldsymbol{W},\boldsymbol{\varXi}) =\frac{|\boldsymbol{\overline{h}}_k^{\rm H}\boldsymbol{\varXi}\boldsymbol{w}_k|^2}
	{\sum_{i\neq k}|\boldsymbol{\overline{h}}_k^{\rm H}\boldsymbol{\varXi}\boldsymbol{w}_i|^2+\frac{\sigma^2}{P_0}\|\boldsymbol W\|_{\rm F}^2}.
\end{align}
Note that the solutions of $(\mathcal P_{\rm P}-2)$ and $(\mathcal P_{\rm P}-3)$ are connected with a scalar. After solving $(\mathcal P_{\rm P}-3)$, we scale its beamforming solution to satisfy the power constraint $P_0$, thereby obtaining a solution to $(\mathcal P_{\rm P}-2)$.

 $(\mathcal P_{\rm P}\text{-}3)$ remains nonconvex because of the WSR objective in \eqref{mbpraa3} and the hardware equality in \eqref{etaCons}. To handle the WSR objective, define the effective channel $\boldsymbol q_k\triangleq\boldsymbol{\varXi}^{\rm H}\overline{\boldsymbol h}_k\in\mathbb C^M$, and let $u_k\in\mathbb C$ and $v_k>0$ denote the receive equalizer and mean-square error (MSE) weight for the $k$th user, respectively. The MSE is
\begin{align}\label{ek}
	e_k\!=\!|u_k|^2\left(\sum_{i=1}^{K}|\boldsymbol q_k^{\rm H}\boldsymbol w_i|^2+\frac{\sigma^2}{P_0}\|\boldsymbol W\|_{\rm F}^2\right)\!-\!2\Re\!\left\{u_k^*\boldsymbol q_k^{\rm H}\boldsymbol w_k\right\}+1.
\end{align}
Let $\boldsymbol u\triangleq[u_1,\ldots,u_K]^{\rm T}$ and $\boldsymbol v\triangleq[v_1,\ldots,v_K]^{\rm T}$. Maximizing \eqref{mbpraa3} is equivalent to minimizing
\begin{align}\label{equivalence}
	\Phi(\boldsymbol W,\boldsymbol{\varXi},\boldsymbol u,\boldsymbol v)\triangleq\sum_{k=1}^{K}\chi_k\left(v_ke_k-\ln v_k\right).
\end{align}

To address the nonconvex hardware equality constraint, we incorporate it into the objective through the penalty term
\begin{align}\label{Mini}
	\mu\sum_{m=1}^{M}\left\|\delta_m\boldsymbol\eta_m-\boldsymbol\zeta_m(\boldsymbol t_m)\right\|_2^2,
\end{align}
where $\mu>0$ is the penalty coefficient and $|\delta_m|=1$. The phase rotation $\delta_m$ permits phase-invariant alignment between the continuous and practical Fourier-coefficient vectors. Because this per-antenna phase can be absorbed into the corresponding row of the digital beamformer, it does not change the achievable system performance. The penalized objective is
\begin{align}\label{Lrho}
	\mathcal L_{\mu}\triangleq\Phi(\boldsymbol W,\boldsymbol{\varXi},\boldsymbol u,\boldsymbol v)+\mu\sum_{m=1}^{M}\left\|\delta_m\boldsymbol\eta_m-\boldsymbol\zeta_m(\boldsymbol t_m)\right\|_2^2.
\end{align}
Consequently, the penalized problem is
\begin{align}
	(\mathcal P_{\rm P}-4)~&\min_{\boldsymbol W,\boldsymbol{\varXi},\overline{\boldsymbol t},\boldsymbol u,\boldsymbol v,\boldsymbol\delta}\quad\mathcal L_{\mu}\nonumber\\
	&\mathrm{~~~~~s.t.}~\eqref{mbprac},~\boldsymbol\eta_m\in\overline{\mathcal T},~v_k>0,~|\delta_m|=1,
\end{align}
where $\boldsymbol\delta\triangleq[\delta_1,\ldots,\delta_M]^{\rm T}$.

Note that for a fixed $\mu$, $(\mathcal P_{\rm P}-4)$ has a tractable block structure when the other variables are fixed. We therefore cyclically update $\boldsymbol u$, $\boldsymbol v$, $\boldsymbol W$, $\boldsymbol{\varXi}$, $\overline{\boldsymbol t}$, and $\boldsymbol\delta$.  Moreover, a finite penalty coefficient may not sufficiently enforce $\boldsymbol{\eta}_m=\boldsymbol{\zeta}_m(\boldsymbol t_m)$, whereas an excessively large initial penalty coefficient may cause the hardware-consistency term to dominate the WMMSE objective and thereby hinder convergence to a high-quality solution. Following the penalty-decomposition principle~\cite{Shi2020PDD}, we therefore initialize $\mu$ with a small positive value and progressively increase it. Accordingly, $(\mathcal P_{\rm P}-4)$ is solved using a two-loop procedure, in which the inner loop alternately updates the variable blocks for a fixed $\mu$, while the outer loop increases $\mu$ until the hardware-consistency residual satisfies a prescribed tolerance.

\subsection{Inner-Loop Variable Updates}
\subsubsection{Optimization of $\boldsymbol{u}$} With all other variables fixed, $(\mathcal P_{\rm P}-4)$ reduces to
\begin{align}\label{uk}
	\min_{u_k}~\sum_{k=1}^{K}\chi_k\left(v_ke_k-\log v_k\right).
\end{align}
From \eqref{uk}, the optimal $u_k$ is
\begin{align}\label{eq:u}
	\widehat u_k=\frac{\boldsymbol q_k^{\rm H}\boldsymbol w_k}{\sum_{i=1}^{K}|\boldsymbol q_k^{\rm H}\boldsymbol w_i|^2+\frac{\sigma^2}{P_0}\|\boldsymbol W\|_{\rm F}^2}.
\end{align}
\subsubsection{Optimization of $v_k$}
With the remaining variables fixed and $u_k=\widehat u_k$, the optimal MSE weight is
\begin{align}\label{eq:v}
	\widehat v_k=1/\widehat e_k,
\end{align}
where $\widehat e_k$ denotes \eqref{ek} evaluated at $u_k=\widehat u_k$.

\subsubsection{Optimization of $\boldsymbol W$}
We next fix the other variables and optimize $\boldsymbol W$. Define
\begin{align}
	\boldsymbol Q&\triangleq[\boldsymbol q_1,\ldots,\boldsymbol q_K],\qquad \boldsymbol G\triangleq\boldsymbol Q^{\rm H}\boldsymbol Q\in\mathbb C^{K\times K}.
\end{align}
Moreover, define
\begin{align}\label{eq:rwmmse-definitions}
	d_k&\triangleq\chi_kv_k|u_k|^2,\qquad c_k\triangleq\chi_kv_ku_k,\qquad \tau\triangleq\frac{\sigma^2}{P_0}\sum_{k=1}^{K}d_k,\nonumber\\
	\boldsymbol D&\triangleq\operatorname{diag}\{d_1,\ldots,d_K\},\qquad \boldsymbol \Omega\triangleq\operatorname{diag}\{c_1,\ldots,c_K\}.
\end{align}
The terms in \eqref{ek} that depend on $\boldsymbol W$ satisfy
\begin{align}\label{eq:W-quadratic-identities}
	\sum_{k=1}^{K}d_k\sum_{i=1}^{K}|\boldsymbol q_k^{\rm H}\boldsymbol w_i|^2&=\operatorname{Tr}\!\left(\boldsymbol W^{\rm H}\boldsymbol Q\boldsymbol D\boldsymbol Q^{\rm H}\boldsymbol W\right),\nonumber\\
	\frac{\sigma^2}{P_0}\sum_{k=1}^{K}d_k\|\boldsymbol W\|_{\rm F}^2&=\operatorname{Tr}\!\left(\boldsymbol W^{\rm H}\tau\boldsymbol I_M\boldsymbol W\right),\nonumber\\
	-2\Re\!\left\{\sum_{k=1}^{K}\chi_kv_ku_k^*\boldsymbol q_k^{\rm H}\boldsymbol w_k\right\}&=-2\Re\!\left\{\operatorname{Tr}\!\left(\boldsymbol{\Omega}^{\rm H}\boldsymbol Q^{\rm H}\boldsymbol W\right)\right\}.
\end{align}
After removing terms independent of $\boldsymbol W$, we have
\begin{align}\label{eq:W-subproblem}
	\min_{\boldsymbol W}\quad&\operatorname{Tr}\!\left[\boldsymbol W^{\rm H}(\boldsymbol Q\boldsymbol D\boldsymbol Q^{\rm H}+\tau\boldsymbol I_M)\boldsymbol W\right]-2\Re\!\left\{\operatorname{Tr}\!\left(\boldsymbol{\Omega}^{\rm H}\boldsymbol Q^{\rm H}\boldsymbol W\right)\right\}.
\end{align}
For $\tau>0$, \eqref{eq:W-subproblem} is a strictly convex quadratic problem with the unique minimizer
\begin{align}\label{eq:W-update}
	\widehat{\boldsymbol W}&=(\boldsymbol Q\boldsymbol D\boldsymbol Q^{\rm H}+\tau\boldsymbol I_M)^{-1}\boldsymbol Q\boldsymbol{\Omega}\nonumber\\
	&\overset{\rm (a)}{=}\boldsymbol Q(\tau\boldsymbol I_K+\boldsymbol D\boldsymbol G)^{-1}\boldsymbol{\Omega},
\end{align}
where $\rm (a)$ follows from the push-through identity~\cite{Boyd2018Introduction}. Note that \eqref{eq:W-update} requires only a $K\times K$ matrix inversion.

\subsubsection{Optimization of $\boldsymbol{\varXi}$}
Stack the continuous Fourier coefficients as
\begin{align}
	\overline{\boldsymbol\eta}\triangleq[\boldsymbol\eta_1^{\rm T},\ldots,\boldsymbol\eta_M^{\rm T}]^{\rm T}\in\mathbb C^{MT},
\end{align}
and define
\begin{align}
	\boldsymbol E_i&\triangleq\operatorname{diag}\{\boldsymbol w_i^*\}\otimes\boldsymbol I_T\in\mathbb C^{MT\times MT},~\boldsymbol a_{k,i}\triangleq\boldsymbol E_i\overline{\boldsymbol h}_k\in\mathbb C^{MT},\nonumber\\
	\boldsymbol\Delta&\triangleq\operatorname{diag}\{\delta_1,\ldots,\delta_M\}\otimes\boldsymbol I_T,\nonumber\\
	\overline{\boldsymbol\zeta}(\boldsymbol{\overline{t}})&\triangleq[\boldsymbol\zeta_1^{\rm T}(\boldsymbol t_1),\ldots,\boldsymbol\zeta_M^{\rm T}(\boldsymbol t_M)]^{\rm T}.
\end{align}
With these definitions, $\overline{\boldsymbol h}_k^{\rm H}\boldsymbol{\varXi}\boldsymbol w_i=\boldsymbol a_{k,i}^{\rm H}\overline{\boldsymbol\eta}$. After removing terms independent of $\overline{\boldsymbol\eta}$, the Fourier-coefficient subproblem~is
\begin{align}\label{eq:eta-subproblem}
	\min_{\overline{\boldsymbol\eta}}\quad&\overline{\boldsymbol\eta}^{\rm H}\boldsymbol A\overline{\boldsymbol\eta}-2\Re\!\left\{\boldsymbol c^{\rm H}\overline{\boldsymbol\eta}\right\}\nonumber\\
	\mathrm{s.t.}\quad&\boldsymbol\eta_m\in\overline{\mathcal T},\quad m=1,\ldots,M,
\end{align}
where
\begin{align}\label{eq:Aeta-beta}
	\boldsymbol A&\triangleq\sum_{k=1}^{K}\chi_kv_k|u_k|^2\sum_{i=1}^{K}\boldsymbol a_{k,i}\boldsymbol a_{k,i}^{\rm H}+\mu\boldsymbol\Delta^{\rm H}\boldsymbol\Delta,\nonumber\\
	\boldsymbol c&\triangleq\sum_{k=1}^{K}\chi_kv_ku_k\boldsymbol a_{k,k}+\mu\boldsymbol\Delta^{\rm H}\overline{\boldsymbol\zeta}(\boldsymbol{\overline{t}}).
\end{align}
The unit-norm constraints in $\overline{\mathcal T}$ form a product of $M$ complex spheres, while the per-order bounds further restrict the feasible Fourier-coefficient amplitudes. Therefore, \eqref{eq:eta-subproblem} constitutes an inequality-constrained optimization problem over a product-sphere manifold. We solve it using a projected Riemannian gradient method. Let 
\begin{align}\label{eq:eta-egrad}
	\boldsymbol g=2(\boldsymbol A\boldsymbol{\overline{\eta}}-\boldsymbol c)
	=[\boldsymbol g_1^{\rm T},\ldots,\boldsymbol g_M^{\rm T}]^{\rm T}
\end{align}
denote the Euclidean gradient. Its projection onto the tangent space of the $m$th complex sphere is
\begin{align}\label{eq:eta-rgrad}
	\boldsymbol{\overline{g}}_m=\boldsymbol g_m-\Re\!\left\{\boldsymbol\eta_m^{\rm H}\boldsymbol g_m\right\}\boldsymbol\eta_m.
\end{align}
For a step size $\kappa>0$, the tentative point $\boldsymbol\eta_m-\kappa\boldsymbol{\overline{g}}_m$ is projected onto $\overline{\mathcal T}$. At each iteration, $\kappa$ is selected via Armijo backtracking, and the projected tentative point is accepted when it satisfies the sufficient-decrease condition. This procedure is repeated until convergence.
\subsubsection{Optimization of $\boldsymbol{t}_m$ and $\delta_m$}

For fixed $\boldsymbol\eta_m$,  $\boldsymbol t_m$ and $\delta_m$ are decoupled across the PRAs. The $m$th subproblem is
\begin{align}\label{eq:t-delta-subproblem}
	\min_{\boldsymbol t_m,\delta_m}~~&
	\|\delta_m\boldsymbol\eta_m-\boldsymbol\zeta_m(\boldsymbol t_m)\|_2^2\nonumber\\
	\mathrm{s.t.}~~~~&[\boldsymbol t_m]_q\in\{0,1\},~~q=1,\ldots,Q,~~|\delta_m|=1.
\end{align}
Since both $\boldsymbol\eta_m$ and $\boldsymbol\zeta_m(\boldsymbol t_m)$ have unit norm, minimizing \eqref{eq:t-delta-subproblem} over $\delta_m$ for a given $\boldsymbol t_m$ gives
\begin{align}\label{eq:delta-update}
	\widehat\delta_m =\dfrac{\boldsymbol\eta_m^{\rm H}\boldsymbol\zeta_m(\boldsymbol t_m)}
	{|\boldsymbol\eta_m^{\rm H}\boldsymbol\zeta_m(\boldsymbol t_m)|}. 
\end{align}
Then, the minimized residual is
\begin{align}\label{eq:phase-invariant-residual}
	2-2|\boldsymbol\eta_m^{\rm H}\boldsymbol\zeta_m(\boldsymbol t_m)|.
\end{align}
Consequently, the optimization of $\boldsymbol{t}_m$ is equivalently given by
\begin{align}\label{eq:t-correlation}
	\max_{\boldsymbol t_m\in\{0,1\}^{Q}} |\boldsymbol\eta_m^{\rm H}\boldsymbol\zeta_m(\boldsymbol t_m)|.
\end{align}

Exhaustive search for \eqref{eq:t-correlation} evaluates all $2^Q$ switch states and is therefore impractical for large $Q$. We instead propose an OBNS method. Given the current state $\boldsymbol t_m$, define the $q$th one-bit neighbor by
\begin{align}\label{eq:one-flip-neighbor}
	\boldsymbol{\widetilde{t}}_m^{(q)}	\triangleq\boldsymbol t_m+
	(1-2[\boldsymbol t_m]_q)\boldsymbol e_q,\qquad q=1,\ldots,Q,
\end{align}
and set $\boldsymbol{\widetilde{t}}_m^{(0)}=\boldsymbol t_m$. $\boldsymbol e_q$ denotes the $q$th basis vector, whose $q$th entry is one and all other entries are zero. It is easy to verify that the $q$th entries of $\boldsymbol{\widetilde{t}}_m^{(q)}$ and $\boldsymbol t_m$ are complementary, whereas all other entries are identical.  The best neighbor is selected according to
\begin{align}\label{eq:one-flip-update}
	q_m^{\star}\in\arg\max_{q\in\{0,1,\ldots,Q\}}
	\left|\boldsymbol\eta_m^{\rm H}
	\boldsymbol\zeta_m(\boldsymbol{\widetilde{t}}_m^{(q)})\right|.
\end{align}
If $q_m^{\star}\neq0$ and the objective is strictly improved, we set $\boldsymbol t_m\leftarrow\boldsymbol{\widetilde{t}}_m^{(q_m^{\star})}$ and repeat \eqref{eq:one-flip-update}; otherwise, the search terminates. Because the current state is included in the candidate set, the objective in \eqref{eq:t-delta-subproblem} does not increase after any accepted flip. When the OBNS method terminates, we denote the resulting switch state by $\widehat{\boldsymbol t}_m$.

One potential drawback of the OBNS method is the repeated evaluation of $\boldsymbol\zeta_m(\boldsymbol{\widetilde{t}}_m^{(q)})$, which involves an inversion of a $Q\times Q$ matrix in \eqref{Ip} and incurs a high computational complexity of $\mathcal{O}(Q^3)$. To accelerate the evaluation, we exploit the diagonal structure of the switch-load perturbation to reduce the computational complexity. Define
\begin{align}
	\boldsymbol J(\boldsymbol t_m)\triangleq
	\boldsymbol Z_{\rm PP}+\boldsymbol Z_{\rm L}(\boldsymbol t_m).
\end{align}
Flipping the $q$th switch changes only one diagonal entry, and thus the candidate loaded-impedance matrix is
\begin{align}
	\boldsymbol J(\boldsymbol{\widetilde{t}}_m^{(q)})=\boldsymbol J(\boldsymbol t_m)+
	\Delta_q \boldsymbol e_q\boldsymbol e_q^{\rm T},
\end{align}
where $\Delta_q$ denotes the corresponding change in the load impedance. According to the Sherman--Morrison identity, we have
\begin{align}\label{eq:woodbury-switch}
		\boldsymbol J(\boldsymbol{\widetilde{t}}_m^{(q)})^{-1}=\boldsymbol J(\boldsymbol t_m)^{-1}-
	\frac{\Delta_{q}\boldsymbol J(\boldsymbol t_m)^{-1}\boldsymbol e_q
		\boldsymbol e_q^{\rm T}\boldsymbol J(\boldsymbol t_m)^{-1}}
	{1+\Delta_{q}\boldsymbol e_q^{\rm T}
		\boldsymbol J(\boldsymbol t_m)^{-1}\boldsymbol e_q}.
\end{align}
Note that $\boldsymbol J(\boldsymbol t_m)^{-1}$ has been obtained in the previous iteration. Therefore, no new matrix inversion is required for $\boldsymbol J(\boldsymbol{\widetilde{t}}_m^{(q)})^{-1}$. In addition, the calculation of $\boldsymbol J(\boldsymbol{\widetilde{t}}_m^{(q)})^{-1}$ involves sparse vectors $\boldsymbol{e}_q$ and thus can be exploited to further improve the calculation efficiency. After the OBNS terminates, $\delta_m$ is updated using \eqref{eq:delta-update}.

\subsection{Outer-Loop Updates and Stop Criteria}

For a fixed penalty coefficient $\mu$, the inner loop cyclically updates the communication block $\{\boldsymbol W,u_k,v_k\}$, the continuous Fourier-coefficient block $\boldsymbol{\varXi}$, and the practical-state block $\{\boldsymbol t_m,\delta_m\}$. The inner loop terminates when the variation of the objective in \eqref{Lrho} between adjacent iterations is less than a predefined threshold $\epsilon_{\rm in}$ or when the prescribed maximum number of inner iterations is reached.

For the outer iteration, define the hardware-consistency residual as
\begin{align}\label{eq:hardware-residual}
	  \varrho\triangleq \sum_{m=1}^{M}\left\|\delta_m\boldsymbol\eta_m-\boldsymbol\zeta_m(\boldsymbol t_m)\right\|_2^2.
\end{align}
If $\varrho>\epsilon_{\rm out}$, the penalty coefficient is increased according to $\mu \leftarrow \beta\mu$, where $\beta$ is the penalty scaling factor. The outer loop terminates when $\varrho\leq\epsilon_{\rm out}$ or when the maximum number of outer iterations is reached.

After the CBP-WMMSE algorithm terminates, we normalize $\widehat{\boldsymbol{W}}$ via
\begin{align}\label{Ww}
	\widetilde{\boldsymbol{W}} = \sqrt{P_0}\widehat{\boldsymbol{W}}/\|\widehat{\boldsymbol{W}}\|_{\rm F}.
\end{align}
Then, we absorb the phase rotation $\delta_m$ into $\widetilde{\boldsymbol{W}}$ and obtain the final digital beamformer as
\begin{align}\label{Phase}
	\overline{\overline{\boldsymbol{W}}} = \mathrm{diag}\{\delta_1^*,\delta_2^*,\cdots,\delta_M^*\}\widetilde{\boldsymbol{W}}.
\end{align}
Finally, we summarize the complete procedure of the CBP-WMMSE algorithm in \textbf{Algorithm~\ref{alg:modal-bound-pdd-rwmmse}}.

\begin{algorithm}[!t]
	\caption{CBP-WMMSE Algorithm}
	\label{alg:modal-bound-pdd-rwmmse}
	\begin{algorithmic}[1]
		\STATE \textbf{Input:} $\overline{\boldsymbol h}_k$, $\chi_k$, $\boldsymbol Z_{\rm PP}$, $\boldsymbol z_{\rm PA}$, $b_q(\theta)$, $\boldsymbol C$, $P_0$, $M$, $Q$, $K$, $\epsilon_{\rm in}$, $\epsilon_{\rm out}$, $\beta$, $I_{\rm in}$, and $I_{\rm out}$.
		\STATE Initialize $\boldsymbol W$, $\boldsymbol\eta_m$, $\boldsymbol t_m$, and $\delta_m$.
		\WHILE{the stop conditions are not satisfied}
		\STATE $\mu\leftarrow\beta\mu$.
		\WHILE{the stop conditions are not satisfied}
		\STATE Update $u_k$ for all $k$ via \eqref{eq:u}.
		\STATE Update $v_k$ for all $k$ via \eqref{eq:v}.
		\STATE Update $\boldsymbol{W}$ via \eqref{eq:W-update}.
		\STATE Update $\boldsymbol{\eta}_m$ for all $m$ by solving \eqref{eq:eta-subproblem}.
		\STATE Update $\delta_m$ and $\boldsymbol{t}_m$ for all $m$ by solving \eqref{eq:t-delta-subproblem}.
		\ENDWHILE
		\ENDWHILE
		\STATE Normalize and update $\boldsymbol{W}$ via \eqref{Ww} and \eqref{Phase}.
		\STATE \textbf{Output:} $\overline{\overline{\boldsymbol{W}}}$ and $\widehat{\boldsymbol{t}}_m$.
	\end{algorithmic}
\end{algorithm}

\subsection{Convergence Analysis}

We first consider the inner loop for a fixed $\mu$. The closed-form updates in \eqref{eq:u}, \eqref{eq:v}, and \eqref{eq:W-update} minimize their respective subproblems and hence do not increase $\mathcal L_{\mu}$. The projected Riemannian gradient update of $\boldsymbol{\varXi}$ satisfies the Armijo condition, while the one-flip update in \eqref{eq:one-flip-update} accepts a switch state only when it decreases \eqref{eq:t-delta-subproblem}. Moreover, \eqref{eq:delta-update} gives the optimal phase rotation for a fixed switch state. Therefore, the objective generated by the inner loop is non-increasing. Since $\mathcal L_{\mu}$ is bounded below, the inner loop is guaranteed to converge. We next consider the outer loop. According to \textbf{Algorithm~\ref{alg:modal-bound-pdd-rwmmse}}, $\mu$ is successively increased by a factor $\beta>1$, thereby progressively strengthening the penalty on the hardware-consistency residual in \eqref{eq:hardware-residual}. Therefore, the hardware-consistency residual approaches zero as $\mu\rightarrow\infty$ and the hardware constraint in \eqref{etaCons} can be satisfied when the  algorithm terminates.
\subsection{Computational Complexity}

Let $I_{\rm out}$ and $I_{\rm in}$ denote the numbers of outer and inner iterations, respectively. The outer loop only updates the scalar penalty coefficient according to \textbf{Algorithm~\ref{alg:modal-bound-pdd-rwmmse}}, and thus its complexity is $\mathcal O(1)$ per iteration.

Within each inner iteration, the CBP-WMMSE algorithm iteratively updates $u_k$, $v_k$, $\boldsymbol{W}$, $\boldsymbol{\eta}_m$, $\delta_m$, and $\boldsymbol{t}_m$. The updates of $u_k$ and $v_k$ mainly involve matrix multiplication and have a computational complexity of  $\mathcal O(KMT + MK^2)$. For the update of $\boldsymbol W$, forming $\boldsymbol G$ costs $\mathcal O(MK^2)$, while \eqref{eq:W-update} requires one $K\times K$ matrix inversion and the associated matrix multiplications. Therefore, the complexity of updating $\boldsymbol W$ is $\mathcal O(K^3+MK^2)$. For the update of $\boldsymbol{\varXi}$,  each gradient in \eqref{eq:eta-egrad} is evaluated in $\mathcal O(K^2MT)$ operations. Denoting the number of projected Riemannian gradient steps by $I_{\rm R}$, this block has complexity $\mathcal O(I_{\rm R}K^2MT)$. For the update of $\delta_m$, it has closed-form solution in \eqref{eq:delta-update} and mainly involve matrix multiplication with computational complexity $\mathcal{O}(T)$. For the update of $\boldsymbol{t}_m$, one round of the OBNS evaluates $Q$ neighbors for each PRA. Using \eqref{eq:woodbury-switch} and caching the current inverse, the complexity of one round is $\mathcal O(Q^2+QT)$ per PRA. Hence, $I_{\rm F}$ one-flip rounds for all PRAs cost $\mathcal O(I_{\rm F}M(Q^2+QT))$.

For the typical setting with $M>\max\{K,T\}$ and $Q>T$, the lower-order terms can be omitted. Combining the above variable updates, the overall computational complexity is
\begin{align}\label{eq:overall-complexity}
	\mathcal O\left(I_{\rm out}I_{\rm in}M
	\left(I_{\rm R}K^2T+I_{\rm F}Q^2\right)
	\right).
\end{align}
 
\subsection{Extensions}
The proposed algorithm can be readily extended to optimization under the ideal PRA model. Specifically, after removing the hardware-consistency constraint \eqref{etaCons} from \eqref{mbpra2} and replacing $\overline{\mathcal T}$ with the ideal feasible set $\mathcal T$ in \eqref{feasibleset},  $\boldsymbol{\eta}_m$ is independent of any specific hardware realization. Thus, $(\mathcal P_{\rm P}-2)$ shares the same problem structure as $(\mathcal P_{\rm I})$. By replacing $\boldsymbol{\eta}_m$ and $\boldsymbol{\varXi}$ with $\boldsymbol{\gamma}_m$ and $\boldsymbol{\Gamma}$, respectively, the remaining updates of $u_k$, $v_k$, $\boldsymbol W$, and $\boldsymbol{\gamma}_m$ provide an algorithm for solving the ideal-model problem $(\mathcal P_{\rm I})$.

More generally, the proposed algorithm establishes a unified framework that separates the original mixed-binary optimization into a continuous beamforming problem and a discrete hardware-approximation problem. The continuous part can incorporate existing beamforming methods developed for different objectives and system constraints~\cite{wang2026sparse}, while the Fourier coefficients $\boldsymbol{\eta}_m$ can be optimized using the projected Riemannian gradient method. The discrete part approximates the optimized Fourier coefficients using realizable switch states and can be efficiently handled by the proposed OBNS method. Since these two parts are connected only through the Fourier coefficients, they can be adapted independently. Therefore, the proposed framework can be extended to different beamforming objectives and PRA implementations without changing its overall decomposition structure.

\section{Simulation Results}\label{SR}

We compare the ideal and practical PRA models and evaluate the proposed CBP-WMMSE algorithm. Unless otherwise stated, the BS employs a half-wavelength uniform linear array with $M=32$ antennas and serves $K=4$ single-antenna users. Each user channel contains $L=6$ paths. The three RPA designs in Fig.~\ref{fig_image3} are considered under both the practical model and their corresponding ideal models. For the position-reconfigurable benchmark, the $M$ antennas are placed within a linear region of length $M\lambda$, subject to a minimum spacing of $\lambda/2$. The conventional antenna array has fixed antenna positions and radiation patterns.

\begin{figure}[!t]
	\centering
	\includegraphics[width=70mm]{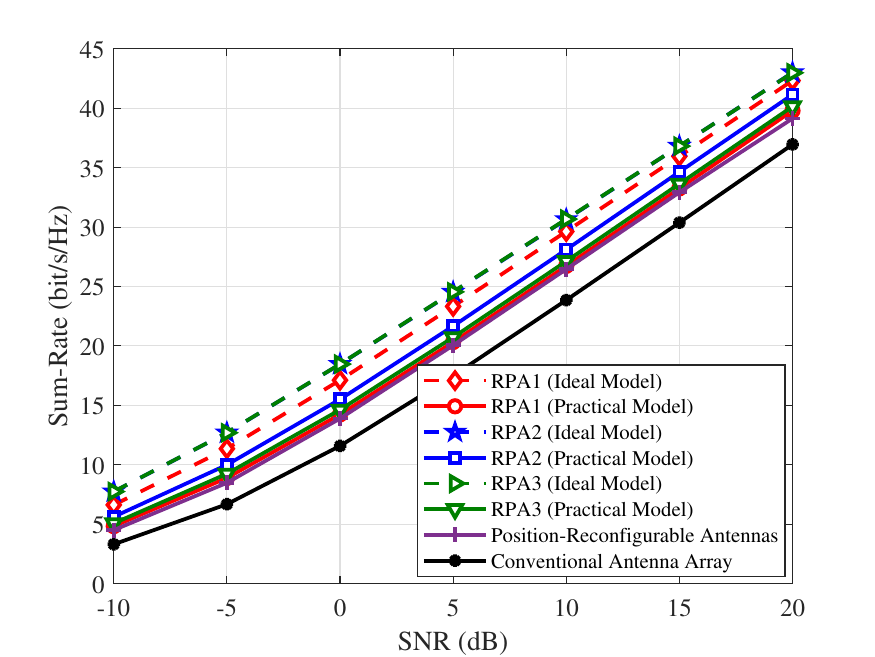}
	\caption{Sum rate versus SNR for different antenna models.}
	\label{fig_image8}
\end{figure}

Fig.~\ref{fig_image8} compares the sum-rates achieved by the different antenna models versus SNR. The ideal PRA models outperform the practical PRA models, which in turn outperform the position-reconfigurable antennas and the conventional antenna array. This ordering reflects the progressively reduced channel-reconfiguration flexibility of these models. The ideal models of RPA~2 and RPA~3 achieve nearly identical performance because the two designs have the same antenna aperture size. Among the practical models, RPA~2 achieves the highest sum rate, followed by RPA~3 and RPA~1. Although RPA~1 has more PIN diodes than RPA~3, its smaller aperture limits its pattern-reconfiguration capability. RPA~2 provides the best performance because it combines a larger aperture with a larger number of PIN diodes.

\begin{figure}[!t]
	\centering
	\includegraphics[width=70mm]{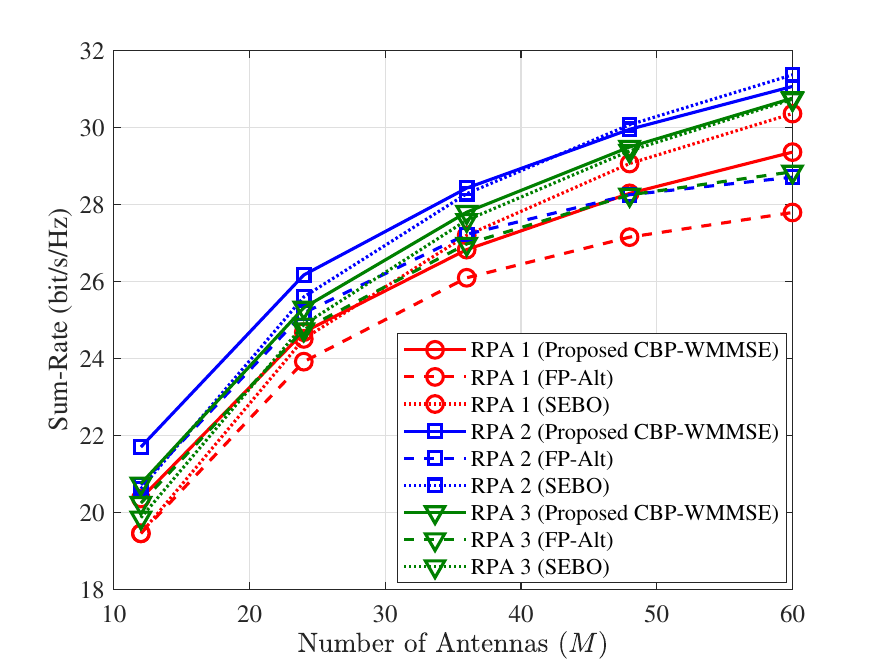}
	\caption{Sum rate versus the number of antennas for different practical-model optimization algorithms.}
	\label{fig_image9}
\end{figure}

\begin{figure}[!t]
	\centering
	\includegraphics[width=70mm]{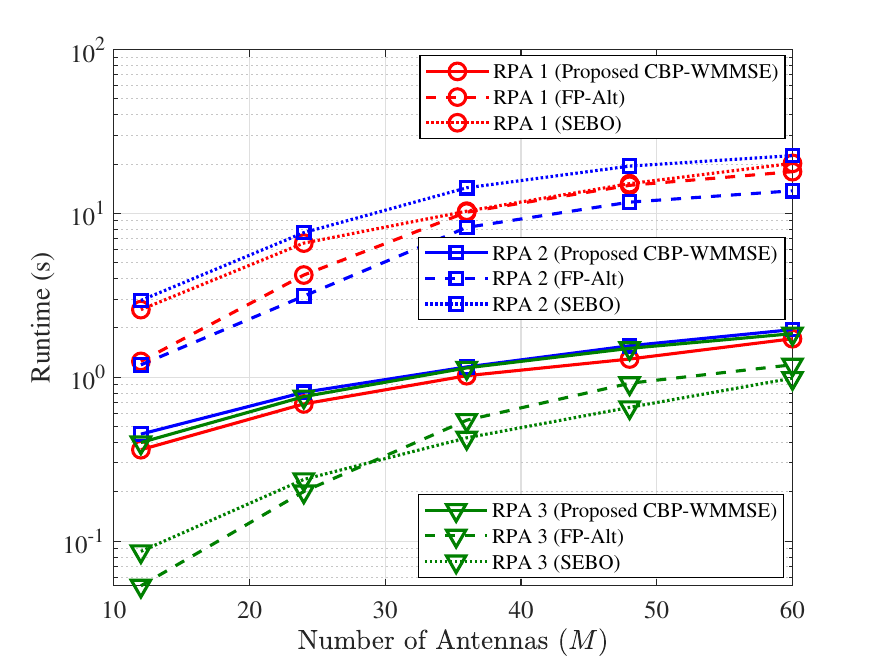}
	\caption{Runtime versus the number of antennas for different practical-model optimization algorithms.}
	\label{fig_image10}
\end{figure}

Figs.~\ref{fig_image9} and~\ref{fig_image10} compare the sum-rate performance and runtime of the proposed algorithm with the fractional-programming-based alternating optimization (FP-Alt) algorithm in~\cite{Li2026AntennaCoding} and the successive exhaustive Boolean optimization (SEBO) algorithm in~\cite{Shen2026AntennaCoding}. The SEBO block size is set to one to reduce the simulation time, and WMMSE beamforming is performed during each SEBO search to maintain high-quality sum-rate solutions. The runtime is measured on an Intel Core Ultra~5 125H processor without parallel computing.

As shown in Fig.~\ref{fig_image9}, the three algorithms achieve generally comparable sum rates, and their curves occasionally intersect as $M$ increases. This is reasonable because the joint beamforming and switch-state design is a nonconvex mixed-binary problem, and the algorithms may reach different locally optimal solutions. Nevertheless, the proposed algorithm provides competitive performance for all tested antenna numbers and RPA designs. In particular, it achieves the highest sum rate for RPA~3 throughout the considered range and provides the best result at most antenna numbers for RPA~2.

Fig.~\ref{fig_image10} shows that, for RPA~3 with only $Q=12$ switches, the proposed algorithm is slower than FP-Alt and SEBO because the additional continuous Fourier-coefficient optimization introduces noticeable overhead. When the number of switches increases to $Q=60$, as in RPA~1 and RPA~2, the advantage of the proposed algorithm becomes evident. Its runtime increases only moderately, whereas the runtimes of FP-Alt and SEBO increase substantially. For RPA~2 with $M=60$, the proposed algorithm requires approximately 1/12 of the runtime of SEBO, demonstrating its efficiency for densely switched RPAs.

\begin{figure}[!t]
	\centering
	\includegraphics[width=70mm]{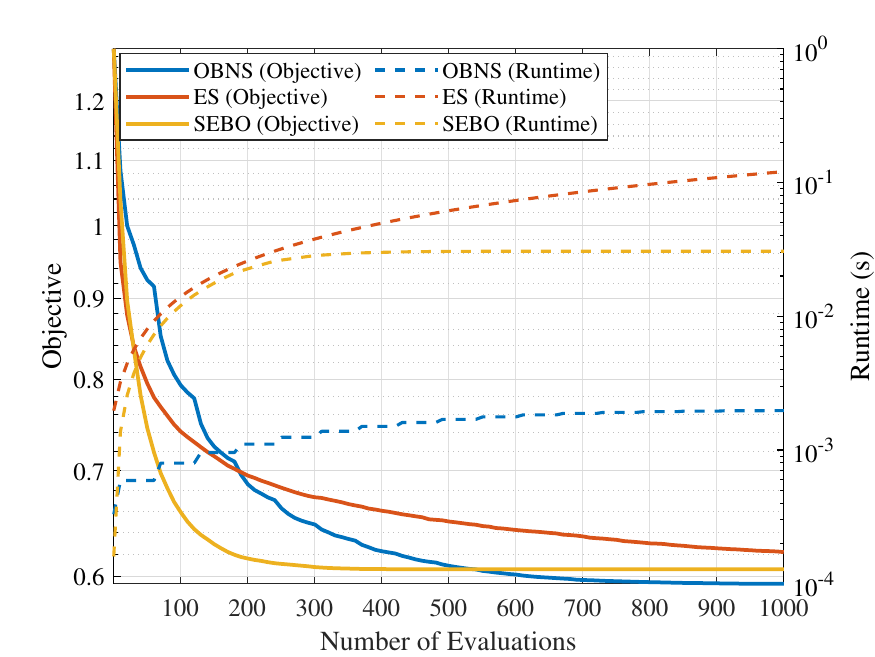}
	\caption{Objective value and runtime versus the number of switch-state evaluations for RPA~2.}
	\label{fig_image11}
\end{figure}

The discrete search is critical to the optimization of practical RPAs. Fig.~\ref{fig_image11} therefore compares the proposed OBNS, SEBO, and exhaustive search (ES) for minimizing the residual in \eqref{eq:phase-invariant-residual} for a single RPA~2. The runtime is measured on the same computing platform described above. OBNS initially decreases the objective more slowly than SEBO because it compares the complete one-bit neighborhood before updating the switch state. However, it eventually attains the lowest objective, indicating that its best-neighbor update leads to a better search trajectory for the considered problem. Moreover, although OBNS evaluates more candidates, its runtime remains substantially lower because it avoids repeated matrix inversions and reuses the inverse and intermediate quantities associated with the current switch state. At $1000$ evaluations, OBNS achieves a lower objective than SEBO and ES, while requiring only approximately $1/15$ and $1/60$ of their runtimes, respectively. These results demonstrate the effectiveness and computational efficiency of the proposed discrete search.

\section{Conclusion}\label{CON}
This paper investigated the modeling and beamforming optimization of PRA-enabled multiuser communication systems. We first interpreted pattern reconfiguration as complex reweighting of multipath components and distinguished it from the phase-only reweighting induced by position reconfiguration. We then developed a finite-aperture ideal PRA model by relating the antenna aperture to the effective Fourier order and deriving per-order power bounds on the Fourier coefficients. A multiport network model was further established for practical RPAs, and full-wave simulations together with three RPA designs verified the consistency between the ideal and practical models. For multiuser downlink transmission, we formulated WSR maximization problems under both models and developed a CBP-WMMSE algorithm for the practical mixed-binary problem. The proposed OBNS method, combined with the Sherman--Morrison update, enables efficient switch-state optimization without repeated matrix inversions. Simulation results showed that PRA-based channel reconfiguration provides clear sum-rate gains over position-reconfigurable and conventional antenna arrays, while the proposed algorithm achieves competitive performance with substantially reduced runtime for RPAs with a large number of switches.

\appendices
\section{Proof of Proposition 1}\label{app:proof1}
The radiation pattern can be expressed as the integral of the radiations from all points within the antenna aperture. Setting the center of the smallest circle enclosing the antenna aperture as the origin, let
$\boldsymbol p= [r\cos\phi,r\sin\phi]^{\rm T}$ denote a point in the aperture and let $\boldsymbol x =[\cos\theta,\sin\theta]^{\rm T}$
denote the observation direction. Therefore, we have
\begin{align}\label{Radiation}
	f(\theta)
	&=
	\int_{\mathcal A}
	W(\boldsymbol p)
	e^{j\frac{2\pi}{\lambda}
		\boldsymbol p^{\rm T}\boldsymbol x}
	{\rm d}\boldsymbol p \nonumber\\
	&=
	\int_{\mathcal A}
	W(\boldsymbol p)
	e^{j\frac{2\pi r}{\lambda}
		\cos(\theta-\phi)}
	{\rm d}\boldsymbol p,
\end{align}
where $\mathcal A$ denotes the antenna-aperture region, $r=\|\boldsymbol p\|_2$, and $W(\boldsymbol p)$ is the current intensity at $\boldsymbol p$. Applying the Jacobi-Anger expansion gives~\cite{Gbur2011}
\begin{align}\label{JacobiAnger}
	e^{j\frac{2\pi r}{\lambda}
		\cos(\theta-\phi)}
	=
	\sum_{n=-\infty}^{\infty}
	j^n
	J_n\left(\frac{2\pi r}{\lambda}\right)
	e^{jn\theta}e^{-jn\phi},
\end{align}
where $J_n(\cdot)$ is the first-kind Bessel function .

Substituting \eqref{JacobiAnger} into \eqref{Radiation}, we obtain
\begin{align}
	f(\theta)
	=&
	\int_{\mathcal A}
	W(\boldsymbol p)
	\sum_{n=-\infty}^{\infty}
	j^n
	J_n\left(\frac{2\pi r}{\lambda}\right)
	e^{jn\theta}e^{-jn\phi}
	{\rm d}\boldsymbol p
	\nonumber\\
	=&
	\sum_{n=-\infty}^{\infty}
	\left[
	j^n
	\int_{\mathcal A}
	W(\boldsymbol p)
	J_n\left(\frac{2\pi r}{\lambda}\right)
	e^{-jn\phi}
	{\rm d}\boldsymbol p
	\right]e^{jn\theta}.
	\label{RadiationFourier}
\end{align}
By comparing \eqref{RadiationFourier} with
\eqref{FourierTransform}, we have
\begin{align}\label{FourierCoefficient}
	\gamma_n
	=
	j^n
	\int_{\mathcal A}
	W(\boldsymbol p)
	J_n\left(\frac{2\pi r}{\lambda}\right)
	e^{-jn\phi}
	{\rm d}\boldsymbol p.
\end{align}
The relationship in \eqref{FourierCoefficient} shows that the
$n$th angular mode is governed by the $n$th-order Bessel function
evaluated over the antenna aperture. We next determine the significant
range of $n$. Define the maximum aperture radius as
\begin{align}
	R_{\mathcal A}
	\triangleq
	\max_{\boldsymbol p \in\mathcal A}~r 
	=
	\max_{\boldsymbol p \in\mathcal A}~
	\|\boldsymbol p\|_2.
\end{align}
Its normalized electrical radius is $2\pi R_{\mathcal A}/\lambda$.
For every point within the antenna aperture, the argument of the
Bessel function satisfies
\begin{align}
	0
	\leq
	\frac{2\pi r}{\lambda}
	\leq
	\frac{2\pi R_{\mathcal A}}{\lambda}.
\end{align}

For an integer $n\geq1$, the Bessel function satisfies~\cite{Olver2010}
\begin{align}
	|J_n(x)|
	\leq
	\frac{1}{n!}\left(\frac{x}{2}\right)^n
	\leq
	\left(\frac{ex}{2n}\right)^n,
	\label{BesselSimpleBound}
\end{align}
where the second inequality follows from
$n!\geq(n/e)^n$. Since
$x\leq2\pi R_{\mathcal A}/\lambda$ over the antenna aperture,
\eqref{BesselSimpleBound} further gives
\begin{align}
	|J_n(x)|
	\leq
	\left(
	\frac{e\pi R_{\mathcal A}}
	{n\lambda}
	\right)^n.
\end{align}
Accordingly, \eqref{FourierCoefficient} implies
\begin{align}\label{gamma_n}
	|\gamma_n|
	\leq
	\left(
	\int_{\mathcal A}|W(\boldsymbol p)|
	{\rm d}\boldsymbol p
	\right)
	\left(
	\frac{e\pi R_{\mathcal A}}
	{n\lambda}
	\right)^n.
\end{align}
When
\begin{align}
	n>\frac{e\pi R_{\rm A}}{\lambda},
\end{align}
${e\pi R_{\mathcal A}}/{(n\lambda)}< 1$, and the second factor on the right-hand side of \eqref{gamma_n} rapidly approaches zero as $n$ increases. For negative orders,
$J_{-n}(x)=(-1)^nJ_n(x)$, and hence the same upper envelope applies
with $n$ replaced by $|n|$. Therefore, a finite antenna aperture induces
a rapidly decaying envelope on sufficiently high-order Fourier coefficients.
This motivates the aperture-dependent effective Fourier truncation order
\begin{align}
	N\triangleq
	\left\lceil\frac{e\pi R_{\rm A}}{\lambda}\right\rceil .
\end{align}


\section{Derivation of $\rho_n$}\label{app:proof2}
From \eqref{FourierCoefficient}, the power of the $n$th Fourier coefficient satisfies
\begin{align}
	|\gamma_n|^2
	&=
	\left|
	\int_{\mathcal A}
	W(\boldsymbol p)
	J_n\left(
	\frac{2\pi r}{\lambda} 
	\right)e^{-jn\phi}
	{\rm d}\boldsymbol p
	\right|^2
	\leq
	P_W B_n,
	\label{GammaCauchyBound}
\end{align}
where the inequality follows from the Cauchy--Schwarz inequality, and
\begin{align}
	P_W
	\!\triangleq\!
	\int_{\mathcal A}|W(\boldsymbol p)|^2{\rm d}\boldsymbol p,\
	B_n
	\!\triangleq\!
	\int_{\mathcal A}
	\left|
	J_n\left(
	\frac{2\pi r}{\lambda}
	\right)
	\right|^2
	{\rm d}\boldsymbol p.
\end{align}
For a given antenna aperture, $B_n$ is a geometry-dependent constant that can be evaluated numerically. Moreover, since
$J_{-n}(x)=(-1)^nJ_n(x)$, we have $B_{-n}=B_n$.

Although the total radiation-pattern power is normalized in
\eqref{CoefficientPowerConstraint}, it does not directly provide a finite upper bound on $P_W$. In particular, weakly radiating or non-radiating current components and destructive current cancellation may produce a large current norm while maintaining the same normalized far-field power. Therefore, obtaining a universal strict upper bound on $P_W$ solely from the radiation-power normalization is generally difficult. To characterize the physically admissible current intensity, we introduce the modeling constraint
\begin{align}\label{CurrentPowerConstraint}
	P_W\leq C,
\end{align}
where $C>0$ is a current-capability parameter. A smaller $C$ represents a more restrictive antenna model with limited current intensity and weaker high-order excitation, whereas a larger $C$ allows stronger and more complex current distributions.  Combining \eqref{GammaCauchyBound} and
\eqref{CurrentPowerConstraint} yields
\begin{align}
	|\gamma_n|^2\leq C B_n.
\end{align}
Meanwhile, the total power constraint
\eqref{CoefficientPowerConstraint} implies
$|\gamma_n|^2\leq1$ for every order. Hence,
\begin{align}\label{IndividualModeBound}
	|\gamma_n|^2
	\leq
	\min\{1,C B_n\}.
\end{align}

\bibliographystyle{IEEEtran}
\bibliography{IEEEabrv,IEEEexample}

\end{document}